%% file: IterationTheory.tex
\documentclass[11pt]{article}
\usepackage{a4}
\usepackage[parfill]{parskip}
\usepackage{amsmath,amssymb,latexsym,stmaryrd}
\usepackage{enumitem}
\usepackage{todonotes}

\usepackage{proof}
\usepackage[all]{xy}
\usepackage{tikz}
\usetikzlibrary{arrows,shapes,automata,cd}
\usetikzlibrary{calc} 
\usetikzlibrary{matrix,arrows} 
\tikzset{
	commutative diagram/.style 2 args={
		matrix of math nodes, row sep=#1,column sep=#2,
		text height=1.5ex, text depth=0.25ex},
	commutative diagram/.default={1cm}{1cm}
}
\tikzset{    
	skip loop/.style n args={3}{to path={-- ++(0,#1) -| node[pos=0.25,#2] {#3} (\tikztotarget)}},
	cross line/.style={preaction={draw=white, -, line width=6pt}}
}

\usepackage{comment}

\usepackage{datetime}
\newdateformat{mydate}{\ordinaldate{\THEDAY}\ \monthname[\THEMONTH]\
	\THEYEAR} 

\include{macros_prakash}

\newcommand{\ang}[1]{\langle #1\rangle}

\newcommand{\e}{\epsilon}
\newcommand{\U}{\mathcal{U}}

\newcommand{\B}{\mathcal{B}}
\newcommand{\A}{\mathcal{A}}

\newcommand{\E}{\mathcal{E}}
\newcommand{\SSS}{\mathbb{U}}

\newcommand{\V}{\mathcal{V}}

\newcommand{\MM}{\mathcal{M}}

\newcommand{\Om}{\hat{\Omega}}

\newcommand{\ol}[1]{\overline{#1}}

\newcommand{\prationals}{\mathbb{Q}_+}
\newcommand{\preals}{\mathbb{R}_+}

\newcommand{\M}{\mathfrak{M}}
\newcommand{\lp}{\llparenthesis}
\newcommand{\rp}{\rrparenthesis}

\newcommand{\Q}{\mathbf{QA}}
\newcommand{\CM}{\mathbf{C}}

\newcommand{\str}[3]{(#1_{#2})_{#2\in #3}}

\newcommand{\ignore}[1]{}

\usepackage{tikz}
\usetikzlibrary{calc,matrix,arrows} 
\tikzset{
	commutative diagram/.style 2 args={
		matrix of math nodes, row sep=#1,column sep=#2,
		text height=1.5ex, text depth=0.25ex},
	commutative diagram/.default={1cm}{1cm}
}
\tikzset{    
	skip loop/.style n args={3}{to path={-- ++(0,#1) -| node[pos=0.25,#2] {#3} (\tikztotarget)}},
	cross line/.style={preaction={draw=white, -, line width=6pt}}
}

\begin{document}
	\bibliographystyle{alpha} %
	
	\title{Fixed-Points for Quantitative Equational Logics}
	\author{
		\begin{tabular}[t]{c}
			Radu Mardare \\
			{\small Department of Computer Science}\\
			{\small University of Aalborg}
		\end{tabular}
		$\qquad$
		\begin{tabular}[t]{c}
			Prakash Panangaden\thanks{Research supported by NSERC, Canada.}\\
			{\small School of Computer Science}\\
			{\small McGill University}
		\end{tabular}
		\\
		\begin{tabular}[t]{c}
			\\
			Gordon Plotkin\\
			{\small School of Informatics}\\
			{\small University of Edinburgh}
		\end{tabular}
	}
	
	\date{\mydate\today}
	
	\maketitle
\begin{abstract}
We develop a fixed-point extension of quantitative equational logic and give
semantics in one-bounded complete quantitative algebras.  Unlike previous related work
about fixed-points in metric spaces, we are working with the notion of
approximate equality rather than exact equality.  The result is a novel theory
of fixed points which can not only provide solutions to the traditional
fixed-point equations but we can also define the rate of convergence to the
fixed point.  We show that such a theory is the quantitative analogue of a
Conway theory and also of an iteration theory; and it reflects the metric coinduction principle.  We study the Bellman equation
for a Markov decision process as an illustrative example. 
\end{abstract}

\section{Introduction}

Quantitative equational logic was introduced in~\cite{Mardare16,Mardare17} as a way of
generalizing the standard concept of equational logic to encompass the concept
of approximate equality.  Essentially, it allows one to use a logical framework
to perform metric reasoning.  The present work is an extension of that formalism
to reason about fixed points of functions.  Fixed point theory is the
mathematical way to understand recursion and iteration~\cite{Scott69,DeBakker71}
and was extensively studied in a partial order setting based ultimately on
Kleene's fixed point theorem~\cite{Kleene52} or some other related fixed-point
theorem like the Knaester-Tarski theorem.  In this paper we develop the metric
version of fixed point theory based on the Banach fixed point theorem, which
says that contractive functions on a bounded complete metric space have \emph{unique}
fixed points.

We follow the categorical axiomatization of fixed-point theories by Simpson and
Plotkin~\cite{Simpson00}, which focusses on the \textit{Conway theories} developed independently by Bloom and Esik \cite{Bloom93} and by Hasegawa \cite{Hasegawa99}. We develop an axiomatization that satisfies quantitative
analogues of their formulations.  We are also able to leverage the completeness
proof from \cite{Mardare16} to obtain a completeness result in our case.  We
also give an axiomatization of fixed-point operators and show how one can reason
about convergence and convergence rates.  We study the relation to a metric
coinduction principle due to Kozen~\cite{Kozen06,Kozen07}: our axiomatization is the metric analogue of \textit{Park induction} and the Kozen coinduction principle is the quantitative version of \textit{Scott induction}, see \cite{Esik95} for a comprehensive presentation of these. Finally we develop
an extended example: the Bellman equation for Markov Decision
Processes~\cite{Puterman94} which plays a central role in reinforcement
learning~\cite{Sutton98}.

We summarize very briefly the formalism introduced in~\cite{Mardare16,Mardare17}.  The
equality symbol $=$ is annotated by a (small) real number $\varepsilon$ so that
one can write \emph{approximate equality} statements of the form:
\(s =_{\varepsilon} t\), where $s,t$ are terms of some theory.  Intuitively, one
thinks of this as meaning that $s$ and $t$ are ``within $\varepsilon$'' of each
other.  The rules of quantitative equational logic are analogous to the rules
for ordinary equational logic except for an infinitary ``continuity'' rule that
allows one to infer $s =_{\varepsilon}t$ from \(s=_{\varepsilon_i}t\) where
the $\varepsilon_i$ converge to $\varepsilon$ from above.  One can then
introduce quantitative algebras which are algebras that have metric structure
and in which all the operations are nonexpansive.  A completeness theorem is
established and it is shown that free algebras can be defined and one can relate
theories to monads on suitable categories of metric spaces.  One of the main
examples given in \cite{Mardare16} is related to spaces of probability
distributions with the Kantorovich metric.

The authors of \cite{Mardare16} have used \emph{extended} metrics: metrics that
can take on infinite values.  We have used $1$-bounded metrics in this paper
instead.  From the topological point of view these are the same: by using the
standard transformation \( d'(x,y) = d(x,y)/(1+d(x,y))\) one can transform the
extended metric $d$ into a $1$-bounded metric \emph{with the same topology}.
Interestingly, under this transformation a contractive function in the
$1$-bounded sense becomes a function that moves all points into the same
connected component in the extended metric sense.

There is a comprehensive study of iteration theories~\cite{Bloom93} which
develops a variety of examples including metric fixed point theories.  We will
comment on this and other interesting related work~\cite{Goncharov18,Kozen06} at
the end of this paper.  For now we remark that other treatments of metric
fixed-point theories are based on the traditional notion of equality and hence
do not allow quantitative reasoning about convergence.  There are a number of
examples from \cite{Mardare16}, such as barycentric algebras, that cannot be
done without the quantitative setting.  We also have new examples such as the
combination of probabilistic choice and nondeterminism.  

In order to carry out our program we are forced to keep track not just of the
fact that functions are contractive but \emph{exactly how contractive they are}
and, furthermore, we need to track this information for each input to the
function.  So the traditional notion of arity needs to be enriched with
quantitative information that we call \emph{Banach patterns}.  The details are,
in some places, intricate but the intuition will be, we hope, clear.  We have
not seen any related work that keeps track of this kind of quantitative
information. 

\section{Notation}

In what follows we will often manipulate tuples of real numbers.  These encode
the contractiveness information that we need in order to be able to define fixed
points, and are useful for managing sets of variables in complex terms.

If $\ol\alpha=\ang{\alpha_1,..,\alpha_n}$, $\ol{\beta}=\ang{\beta_1,..,\beta_m}$
are tuples for $n\geq 1$ and $i\leq n$, let $|\ol\alpha|=n$ and we use the following notations
\\$\ol\alpha\setminus
i=\ang{\alpha_1,..\alpha_{i-1},\alpha_{i+1},..\alpha_n}$,
\\for $x\in\reals$, $\ol\alpha[x/i]=\ang{\alpha_1,..\alpha_{i-1},x,\alpha_{i+1},..\alpha_n}$
and
\\$\ol\alpha[\ol\beta/i]=\ang{\alpha_1,..\alpha_{i-1},\beta_1..\beta_m,\alpha_{i+1},..\alpha_n}$. 
\\ If we have a tuple $\ol\alpha$, we denote its $i$-th component by $\alpha_i$.

Let $\mathbb U_n$ denote the set of all tuples
$\ol\alpha=\ang{\alpha_1..\alpha_n}\in[0,1]^n$ s.t. $\displaystyle\sum_{1\leq
  i\leq n}\alpha_i\leq 1$. And let $\displaystyle\mathbb
U=\bigcup_{i\geq 0}\mathbb U_i$.  

For arbitrary $\ol\alpha,\ol\alpha^1,\ldots,\ol\alpha^n, \ang{\lambda_1,..,\lambda_n}\in\mathbb U_n$ and $r\leq 1$, we define the following operations:
\begin{enumerate}
	\item Scalar multiplication. $r\ol\alpha=\ang{r\alpha_1,\ldots, r\alpha_n}$
	\item Subconvex
      sum. $\sum_i\lambda_i\ol\alpha^i=\ang{\sum_i\lambda_i\alpha_1^i,..,\sum_i\lambda_n\alpha_n^i}$ 
	\item Contraction. For $i<j$, $\ol\alpha[i<j]=(\ol\alpha\setminus j)[\alpha_i+\alpha_j/i]$.
	\item Iteration. For $i\leq n$ s.t. $\alpha_i<1$, $\mu i.\ol\alpha=\frac{1}{1-\alpha_i}(\ol\alpha\setminus i)$.
\end{enumerate}



 \subsection{Banach patterns}

In what follows we introduce the concept of \textit{Banach pattern} that will be used to characterize nonexpansive functions on metric spaces.
Recall that if $(A,d^A)$ and $(B,d^B)$ are metric spaces, then $f:(A,d^A)^n\to(B,d^B)$ is a nonexpansive function if for arbitrary $\ang{a_1,..,a_n}, \ang{b_1,..,b_n}\in A^n$, 
$$d^B(f(a_1..a_n),f(b_1..b_n))\leq\max_{i\leq
	n}d^A(a_i,b_i).$$  
\begin{definition}
  Let $f:(A,d^A)^n\to(B,d^B)$ be a function between two metric
  spaces. $f$\textit{ admits Banach patterns} if there exists a set $\theta\subseteq_{\textit{fin}}\mathbb U_n$ such
  that for any $\ang{a_1..a_n}, \ang{b_1..b_n}\in A^n$,
  $$d^B(f(a_1..a_n),f(b_1..b_n))\leq\max_{\ol\alpha\in\theta}\sum_{i\leq
    n}\alpha_i d^A(a_i,b_i).$$
  In this case, $\theta$ is a \textbf{Banach pattern} for $f$, and we write $f:n:\theta$.
\end{definition}

\begin{example}
  Let $(M,d)$ be a
  $1$-bounded metric space and $\Delta(M,d)$ the space of Borel probability
  distributions on $(M,d)$ metrized with the Kantorovich metric \\$K^d:\Delta(M,d)^2\to[0,1]$.
	
	Consider, for $\e\in [0,1]$, the barycentric operation on $\Delta(M,d)$,
    $+_\e:\Delta(M,d)^2\to\Delta(M,d)$ defined for arbitrary
    $\mu,\nu\in\Delta(M,d)$ by $$\mu+_\e \nu =\e\mu+(1-\e)\nu.$$ 
	In \cite{Mardare16} it has been demonstrated that for arbitrary
    $\mu,\mu',\nu,\nu'\in\Delta(M,d)$,  
	$$K^d(\mu+_\e\mu',\nu+_\e\nu')\leq \e K^d(\mu,\mu')+(1-\e)K^d(\nu,\nu'),$$
	hence, $+_\e$ has Banach pattern the singleton $\{\ang{\e,1-\e}\}$.
  \end{example}	
\begin{example}
  For another example where the pattern is not a singleton we consider the
  non-deterministic choice function on $\Delta$,
  $\oplus:\Delta(M,d)^2\to H(\Delta(M,d))$, where for a metric space $X$, $HX$
  denotes the space of compact subsets equipped with the Hausdorff metric.  The
  function $\oplus$ is nonexpansive in the Hausdorff metric
  defined for $K^d$, \cite{Mardare16}.  Being nonexpansive in this sense, this
  function satisfies for arbitrary $\mu,\mu',\nu,\nu'\in\Delta(M,d)$,
	$$K^d(\mu\oplus\mu',\nu\oplus\nu')\leq \max\{K^d(\mu,\nu),K^d(\mu',\nu')\}.$$
	In this case the Banach pattern is not a singleton, but we have
    $\oplus:2:\{\ang{0,1},\ang{1,0}\}$. 
  \end{example}	
  \begin{example}
    For a third example, we consider, the
    function $$f:\Delta(M,d)^3\to\Delta(M,d)$$ defined, for arbitrary
    $\mu,\nu,\eta\in\Delta(M,d)$ by $$f(\mu,\nu,\eta)=(\mu+_\e \nu)\oplus\eta,$$
    for some $\e\leq 1$.  We note that for arbitrary
    $\mu,\nu,\eta,\mu',\nu',\eta'\in\Delta(M,d)$,
	$$K^d(f(\mu,\nu,\eta),f(\mu',\nu',\eta'))$$ $$\leq \max\{\e
    K^d(\mu,\mu')+(1-\e)K^d(\nu,\nu'), K^d(\eta,\eta')\},$$ and in this case we
    have $f:3:\{\ang{\e,1-\e,0},\ang{0,0,1}\}$. 
\end{example}

Observe that a function $f:(A,d^A)^n\to(B,d^B)$ is nonexpansive iff it admiths Banach patterns. Indeed, if $f$ is nonexpansive, then $$\{\ang{1,0,..0}, \ang{0,1,0..,0},..,\ang{0,..,0,1}\}\subseteq U_n$$ is a Banach pattern for it, the one encoding exactly the nonexpansiveness property. And reverse, if $f$ admits a Banach pattern $\theta\subseteq_{\textit{fin}}\mathbb U_n$, then nonexpansivess derives from  
$$\max_{\ol\alpha\in\theta}\sum_{i\leq n}\alpha_i d^A(a_i,b_i)\leq \max_{i\leq n} d^A(a_i,b_i).$$ However, often a Banach pattern brings more information about the nonexpansiveness of a function.

We will add Banach patterns to the algebraic signatures over
the category of metric spaces when we will define quantitative algebras with
fixed points.  

It is useful to define some operations on patterns, in addition to the set
theoretic operations. Let $\theta,\theta^1..\theta^n\subseteq\mathbb U_n$,
$\lambda\leq 1$ and $\ang{\lambda_1,..,\lambda_n}\in\SSS_n$. 
\begin{enumerate}
	\item Scalar
      multiplication. $\lambda\theta=\{\lambda\ol\alpha\mid\ol\alpha\in\theta\}\subseteq\SSS_n$. 
	\item Subconvex sum. $\displaystyle\sum_{i\leq
        n}\lambda_i\theta^i=\{\sum_{i\leq
        n}\lambda_i\ol\alpha^i\mid\ol\alpha^i\in\theta_i\}$. 
	\item Contraction. $\theta[i<j]=\{\ol\alpha[i<j]\mid\ol\alpha\in\theta\}\subseteq\SSS_{n-1}$.
	\item Composition. For
      $\zeta_1..\zeta_n\in\SSS_m$, $$\theta\circ\ang{\zeta_1..\zeta_n}=\{\sum_{i\leq
        n}\alpha_i\ol\beta^i\mid\ol\alpha\in\theta,
      \ol\beta^i\in\zeta_i\}\subseteq\SSS_m.$$ 
	\item Fixed point. If for all $\ol\alpha\in\theta$, $\alpha_i<1$, let $$\mu
      i.\theta=\{\mu i.\ol\alpha\mid\ol\alpha\in\theta\}\subseteq\SSS_{n-1}.$$ 
\end{enumerate}
Whenever $\theta$ satisfies $[\forall\ol\alpha\in\theta$, $\alpha_i<1]$, we say
that $\theta$ is \textit{i-contractive} and denote this by $\theta\triangleright
i$. 

We also generalize the notation we introduced for tuples and for
$\theta\subseteq\SSS_n$, $\zeta\subseteq\SSS_m$ and $i\leq n$, let
\\$\theta\setminus i=\{\ol\alpha\setminus
i\mid\ol\alpha\in\theta\}\subseteq\SSS_{n-1}$  and
\\$\theta[\zeta/i]=\{\ol\alpha[\ol\beta/i]\mid\ol\alpha\in\theta,
\ol\beta\in\zeta\}\subseteq\SSS_{n+m-1}$. 


\section{Quantitative Equational Reasoning}\label{quant.eq.th}

In this section we recall the main concepts of quantitative equational reasoning
and quantitative algebras~\cite{Mardare16}. 

\subsection{Quantitative Equational Theory}

We start with a \emph{signature} $\Omega$, which is a set of function symbols of
finite arity (constants have arity $0$).  We write $f:n\in\Omega$ for a function
$f$ of arity $n\geq 0$.

Given a set $X$, let $\Om X$ be the $\Omega$-\emph{algebra generated by} $X$,
i.e., the set of all terms constructed on top of $X$ by using the functions in
$\Omega$.  Note that this set comes already equipped with the structure of an
$\Omega$-algebra.  

For a set $X$ of \emph{variables}, one defines \emph{quantitative
  equations}\footnote{In~\cite{Mardare16} quantitative equations are defined for
  $\e\in\prationals$.  We chose to avoid this restriction here in order to get a
  simpler axiomatization. However, all these developments work properly if we
  restrict to rational indices.} over $\Om X$, which have the form $t=_\e s$ for
$t,s\in\Om X$ and $\e\in\preals$.  We use $\mathcal E(\Om X)$ to denote the set
of quantitative equations on $\Om X$.

Let $\mathcal J(\Om X)$ be the class of \emph{quantitative judgements} on $\Om
X$, which are constructions of the form 
$$\{s_i=_{\e_i}t_i\mid i\in I\}\vdash s=_\e t,$$ where $I$ is
a countable (possible empty) index set, $s_i, t_i,s,t\in\Om X$ and
$\e_i,\e\in\preals$ for all $i\in I$.\\If
$\Gamma\vdash\phi\in\mathcal J(\Om X)$, where $\Gamma\subseteq\mathcal E(\Om X)$ and
$\phi\in\mathcal E(\Om X)$, we refer to the elements of $\Gamma$ as the
\emph{hypotheses} and to $\phi$ as the \textit{conclusion} of the
quantitative judgement.

\begin{definition}[Quantitative Equational Theory] \label{def:QEtheory}
	Given a signature $\Omega$ and a set $X$ of variables,
	the \emph{deductive closure} of a set $\U$ of quantitative judgements on
	$\Om X$ is the smallest set $\ol\U$ of quantitative judgements on $\Om X$
	such that $\U\subseteq\ol\U$, and for arbitrary $t,s \in \Om X$,
	$\e,\e'\in\preals$, $f:|I|\in\Omega$, 	$\Gamma,\Theta\subseteq\mathcal \E(\Om X)$ and $\ol s=\str{s}{i}{I}, \ol t=\str{t}{i}{I}\subseteq\Om X$ and any substitution $\sigma$
	%
	\begin{align*} 
		\text{\textbf{(Refl)}} \quad 
		& \vdash t =_0 t\in\ol\U \,, \\
		\text{\textbf{(Symm)}} \quad 
		& \{t=_\e s\} \vdash s=_\e t\in\ol\U \,, \\
		\text{\textbf{(Triang)}} \quad 
		& \{t =_\e u, u =_{\e'} s \} \vdash t =_{\e+\e'} s\in\ol\U \,, \\
		\text{\textbf{(Max)}} \quad 
		& \{t=_\e s\} \vdash t=_{\e+\e'}s\in\ol\U \,, \text{ for all $\e'>0$} \,, \\ 
		\text{\textbf{(NExp)}} \quad
		& \{t_i=_\e s_i\mid i\in I\} \vdash f(\ol t) =_\e f(\ol s)\in\ol\U \,;
	\end{align*}
	%
	and $\ol\U$ is closed under the following rules
	%
	\begin{align*}
		\text{\textbf{(Cont)}} \quad 
		& \frac{\Gamma\vdash s=_{\e'}t \text{ for all }\e'>\e}{\Gamma\vdash s=_\e t} \,, \\
		\text{\textbf{(Subst)}} \quad
		&\frac{\Gamma \vdash t =_\e s}{\sigma(\Gamma) \vdash \sigma(t) =_\e \sigma(t)} \,, \\
		\text{\textbf{(Assumpt)}} \quad
		& \frac{t =_\e s \in\Gamma}{\Gamma \vdash t =_\e s}\,, \\
		\text{\textbf{(Cut)}} \quad 
		& \frac{\Theta \vdash t =_\e s,~ \Gamma \vdash\Theta}{\Gamma \vdash t =_\e s}\,.
	\end{align*}
	where $\Gamma\vdash\Theta$ means that $\Gamma\vdash \phi$ for all
    $\phi\in \Theta$. 
	A \emph{quantitative equational theory of signature $\Omega$ over $X$} is a set
	$\U$ of quantitative judgements on $\Om X$ such that $$\U=\ol\U.$$  
\end{definition}

\begin{df}[Quantitative Algebra]\label{QA}
	Given a signature $\Omega$, a \emph{quantitative algebra over $\Omega$} is a tuple
	$\A=(A,\Omega^\A,d^\A)$, where $(A,\Omega^\A)$ is an algebra of
		signature $\Omega$, 
		$(A,d)$ is a metric space and any $f:|I|\in\Omega$ is nonexpansive.
\end{df}

A \emph{homomorphism of quantitative algebras} is a \emph{non-expansive}
$\Omega$-homomorphism (of $\Omega$-algebras). 

Given a quantitative algebra $\A=(A,\Omega^\A,d^\A)$ of signature $\Omega$ and a
set $X$ of variables, an \emph{assignment} on $\A$ is a function
$\alpha:X\to A$.  It can be canonically extended to a homomorphism of
$\Omega$-algebras $\alpha:\Om X\to A$ by defining, for any $f:|I|\in\Omega$ and
any $\str{t}{i}{I}\subseteq\Om X$,
	$$\alpha(f(\str{t}{i}{I}))=f^\A((\alpha(t_i))_{i\in I}).$$ 
We denote by $\Omega[X|\A]$ the set of assignments on $\A$.

\begin{df}[Satisfaction]
	Let $\A=(A,\Omega^\A,d^\A)$ be an $\Omega$-quantitative algebra and $\{s_i=_{\e_i}t_i\mid i\in I\}\vdash s=_\e t$ a quantitative judgement on $\Om X$. $\A$ \emph{satisfies} this quantitative judgement \textit{under the
		assignment $\alpha\in\Omega[X|\A]$}, written
	$$\{s_i=_{\e_i}t_i\mid i\in I\}\models_{\A,\alpha} s=_\e t,$$ if
	$[\forall i\in I,d^\A(\alpha(t_i),\alpha(s_i))\leq \e_i]$
		implies $d^\A(\alpha(s),\alpha(t))\leq\e$.	
	\\ We say $\A$ \emph{satisfies the quantitative judgement}, or it is a
	\emph{model} of the quantitative judgement, written  $$\{s_i=_{\e_i}t_i\mid
	i\in I\}\models_{\A} s=_\e t,$$  if $$\forall\alpha\in\Omega[X|\A],
    ~\{s_i=_{\e_i}t_i\mid i\in I\}\models_{\A,\alpha} s=_\e t.$$  
\end{df}

Similarly, for a set of quantitative judgements (or a quantitative equational
theory) $\U$, we say that $\A$ is a model of $\U$ if $\A$ satisfies every
element of $\U$; for simplifying notation we denote this by $\A\models\U$. Let
$\Q(\U)$ denote the set of models of $\U$.  If $\M$ is a set of
$\Omega$-quantitative algebras and $\Gamma\vdash\phi\in\mathcal J(\Om X)$, we
write $\Gamma\models_{\M}\phi$ whenever $\Gamma\models_{\A}\phi$ for all
$\A\in\M$.  In \cite{Mardare16} the following completeness result is
established.
\begin{theorem}[Completeness]\label{completenessQA}
	Given a quantitative  equational theory $\U$ over $\Om X$, 
	\[\Gamma\models_{\Q(\U)}\phi~~\text{ iff }~~\Gamma\vdash\phi\in\U.\]
\end{theorem}

\subsection{Limits in quantitative theories}

Although not explicit in~\cite{Mardare16}, quantitative equational theories have built in the mechanism for equational
reasoning about convergence: this will be useful to us.

\begin{definition}
In general, given a quantitative equational theory $\U$ over $\Om X$, we say
that a sequence $(s_i)_{i\geq 1}\subseteq\Om X$ is \emph{convergent} in $\U$ if
there exists $s\in\Om X$ such that
$$\forall\e>0~\exists k~\forall i\geq k,~\vdash s_i=_\e s\in\U.$$
We say that $s$ is a \emph{limit} of the sequence $(s_i)_{i\geq 1}$.
\end{definition}
It is easy to prove the following using (Triang), (Symm) and (Cont).
\begin{proposition}\label{limit}
	Let $\U$ be a quantitative equational theory over $\Om X$. If the sequence
    $(s_i)_{i\geq 1}\subseteq\Om X$ is convergent in $\U$ and it has both
    $s\in\Om X$ and $t\in\Om X$ as limits, then $$\vdash s=_0t\in\U.$$  
\end{proposition}

This motivates us to denote the limit of the sequence $(s_i)_{i\geq 1}$ by
$\lim_{i}s_i$. 

We can construct convergent sequences of terms by applying non-expansive
functions to convergent sequences. 
\begin{lemma}\label{limit1}
	Let $\U$ be a quantitative equational theory over $\Om X$, $f:n,g:m\in\Om X$ and  $(s_k)_{k\geq 1}\subseteq\Om X$ be a convergent sequence in $\U$. Then for $\ol x\subseteq X^n$ and $\ol y\in X^m$,
	\\(1). $(g(\ol y[s_k/j]))_{k\geq 1}$ is a convergent sequence in $\U$ and $$\vdash \lim_{k}g(\ol y[s_k/j])=_0 g(\ol y[\lim_{k}s_k/j])\in\U.$$
	(2). $(f(\ol x[g(\ol y[s_k/j])/i]))_{k\geq 1}$ is convergent in $\U$ and 
	$$\vdash \lim_k f(\ol x[g(\ol y[s_k/j])/i])=_0 f(\ol x[\lim_k g(\ol y[s_k/j])/i])\in\U.$$
\end{lemma}

\begin{proof}
	(1). Let $s=\displaystyle\lim_k s_k$. Hence, $\forall\e>0~\exists p~\forall
    k\geq p$, \\$\vdash s_k=_\e s\in\U$.	Applying (NExp), we get that
    \\$\forall\e>0~\exists p~\forall i\geq p$, $\vdash g(\ol y[s_k/j])=_\e g(\ol
    y[s/i])\in\U$, i.e., \\$\vdash \lim_{k}g(\ol y[s_k/j])=_0 g(\ol
    y[\lim_{k}s_k/j])\in\U.$ 
	\\(2). After observing that $f(\ol x[g(\ol y)/i])$ is nonexpansive, we
    conclude, as above, that $(f(\ol x[g(\ol y[s_k/j])/i]))_{k\geq 1}$ is
    convergent in $\U$. Next, we apply (1) and prove that \\$\vdash\lim_k f(\ol
    x[g(\ol y[s_k/j])/i])=_0 f(\ol x[g(\ol y[s/j])/i])\in\U$ and \\$\vdash f(\ol
    x[\lim_k g(\ol y[s_k/j])/i])=_0 f(\ol x[ g(\ol y[s/j])/i])\in\U$. \\
    (Triang) concludes the proof. 
\end{proof}
These are easy proofs, the point of including them is to show that standard
facts about the continuity of nonexpansive functions can be stated and proved
within the framework of quantitative equational logic.


\section{Banach Quantitative Theories}

In this section we identify a particular class of quantitative equational
theories that we will call Banach theories.  Later we will see that the Banach
theories are the ones for which we can define fixed-point operators.  

A \emph{Banach signature} $\Omega$ is a signature that assigns to each function symbol $f$ an arity $n\in\mathbb N$ and a Banach pattern $\theta\in\SSS_n$; we write $f:n:\theta$.  In particular, for constants
$c\in\Omega$, we have that $c:0:\{\ang{0}\}\in\Omega$. 

We extend the concept of Banach pattern from the elements of a Banach signature
$\Omega$ to all the terms of $\Om X$ by defining, for arbitrary
$f(\ol x)\in \Om X$ with $\ol x\in X^n$ and $f:n:\theta$; and any $g_1(\ol y)$,
 $\ldots, g_n(\ol y)\in\Om X$ with $\ol y\in X^m$ and $g_i:m:\zeta_i$ for $i\leq n$,
the following Banach patterns for contraction and term composition. 
\begin{enumerate}
	\item If $i<j\leq n$ and $h(\ol x\setminus j)=f(\ol x[x_i/j])$, then $$h:n-1:\theta[i<j].$$
	\item If $h(\ol y)=f(g_1(\ol y),..,g_n(\ol y))$, then $$h:m:\theta\circ\ang{\zeta_1..\zeta_n}.$$
\end{enumerate}
With this definition, we will write $t:n:\theta\in\Om X$ to describe any $n$-ary
term with Banach pattern $\theta$ that can be defined in $\Om X$.  If, in
addition $\theta\triangleright i$, we write
$$f:n:\theta\triangleright i\in\Om X.$$
The reader might usefully think of these definitions first in the case where the Banach patterns are all singletons, in which
case these formulas can be seen as a quantitative analogue of how composition
would be defined in operads (multicategories).

\begin{definition}[Banach closure]
	Consider a quantitative equational theory $\U$ over a set $X$ of variables
    and a Banach signature $\Omega$. The \emph{Banach closure} of $\U$ is the
    smallest quantitative equational theory $\U^B$ that contains $\U$ together
    with the axiom 
	\begin{align*}
		\text{\textbf{(1-bound)}} \quad 
		& \vdash x=_1 y\,,
	\end{align*}
	and it is closed under the following rule stated for arbitrary $\e_i\leq 0$ for $i\leq n$.
	\begin{align*}
		\text{\textbf{(Banach)}} \quad 
		& \frac{f:n:\theta\in\Omega}{\{x_i=_{\e_i}y_i\mid i\leq n\}\vdash
       f(x_1\ldots x_n)=_{\delta}f(y_1\ldots y_n)}\,,
	\end{align*}
	where $\displaystyle\delta=\max_{\ol\alpha\in\theta}\sum_{i\leq n}\alpha_i\e_i$.
\end{definition}


\begin{definition}[Banach theory]
	A quantitative equational theory $\U$ over $\Om X$ is a \emph{Banach theory}
    if $\Omega$ is a Banach signature and $$\U=\U^B.$$ 
\end{definition}

The following two results guarantee that the way we defined the patterns for
composition and contraction respects the Banach rule. 

\begin{lemma}\label{Banach}
	Let $\U$ be a Banach theory over $\Om X$ and $t:n:\theta\in\Om X$. Then, for
    arbitrary $\e_i\geq 0$ for $i\leq n$, 
	$$\{x_i=_{\e_i}y_i\mid i\leq n\}\vdash t(x_1\ldots x_n)=_{\delta}t(y_1\ldots
    y_n)\in\U,$$ where $\displaystyle\delta=\max_{\ol\alpha\in\theta}\sum_{i\leq
      n}\alpha_i\e_i$. 
\end{lemma}

\begin{corollary}
	Let $\U$ be a Banach theory over $\Om X$ and
    $\A=(A,\Omega,d)\in\Q(\U)$. Then, for any term $t:n:\theta\in\Om X$ and any
    $\ol a, \ol b\in A^n$,
    $$d(t^\A(\ol a),t^\A(\ol b))\leq\max_{\ol\alpha\in\theta}\sum_{i\leq n}\alpha_i d(a_i,b_i).$$ 
\end{corollary}


\section{Quantitative Fixed-Point Judgements}

In this section we show how one can add fixed-point operators, which are
essentially second-order constructions, to any Banach theory.
\begin{df}
 Let $\Omega$ be a Banach signature and $X$ a set of variables. The
 \emph{fixed-point extension} of $\Om X$ is the set  
 $$\Om^\mu X=\bigcup_{i\geq 0}\Omega_i,$$ where $\Omega_i$ is defined inductively on $i\geq 0$ as follows:
 \\$\Omega_0=\Om X,$
 \\$\Omega_{k+1}=\{\mu i.f:(n-1):\mu i.\theta~\mid~ f:n:\theta\triangleright i\in\Omega_k\}.$
\end{df}  
Let $\mathcal J(\Om^\mu X)$ be the set of judgements on $\Om^\mu X$, i.e.,
judgements involving quantitative equations between terms in $\Om^\mu X$. In
this way we can speak of quantitative equational theories over $\Om^\mu X$,
respecting the requirements of Definition~\ref{def:QEtheory}.
 
\begin{definition}[Fixed-point extension of Banach theory]
 Given a Banach theory $\U$ over $\Om X$, its \emph{fixed-point extension}
 $\U^\mu$ is the smallest quantitative equational theory over $\Om^\mu X$ that
 contains $\U$ and it is closed under the fixed-point approximation rule
 (Approx) stated below for arbitrary $t,u\in\Om^\mu X$, $\ol s\in (\Om^\mu
 X)^n$, and $\e\geq 0$. 
\begin{align*}
	\text{\textbf{(Approx)}} \quad 
	& \frac{t:n:\theta\triangleright i\in\Om^\mu X}{u=_\e t(\ol s[u/i])\vdash u=_{\frac{\e}{1-a}}(\mu i.t)(\ol s\setminus i)} \,,
\end{align*}
where $a=\max\{\alpha_i\mid \ol\alpha\in\theta\}$.
\end{definition}
Note that since $t:n:\theta\triangleright i$, $a<1$. 

When we take a fixed point, the resulting function of the remaining
arguments may not permit further fixed point operations to be performed.  The
Banach patterns allows us to track exactly when we can and cannot take further
fixed points.

\textbf{Notation.} To simplify the presentation in what follows, it is useful to
adopt a syntactic convention that will allow us to focus on certain variables in
terms with many variables, while treating the rest of them as parameters.  If
$f(\ol x)\in\Om^\mu X$ is a function of arity $n$ with free variables
$\ol x=\ang{x_1..x_n}$, and we need to focus on its $i$-th variable $x_i$, we
write $f\lp x_i \rp$. For instance if $s\in\Om^\mu X$, $f\lp s\rp$ denotes $f(\ol x[s/i])$.
Similarly, if the focus is on two variables, say $x_i,x_j$ for $i<j\leq n$, we
write $f\lp x_i,x_j\rp$. We will use this notation in what follows any time
there is no danger of confusion.  It will allow us to avoid carrying
extra variables around in the syntax.

Given a Banach signature $\Omega$ and a set $X$ of variables, the concept of
iteration of a function on its $i$-th variables, $i\leq n$, can be introduced
for an arbitrary function $f:n\in\Om^\mu X$.  Let $\ol x=\ang{x_1\ldots x_n}\in X^n$
and $s\in\Om^\mu X$. We define:

$$[f]_i^1(\ol x[s/i])=f(\ol x[s/i]),$$
$$[f]_i^{k+1}(\ol x[s/i])=f(\ol x[[f]_i^k(\ol x[s/i])/i]).$$

With the previous notation, we can denote the $k$-th iteration on the $i$-th
variable of $f$ on $s$ by $[f]_i^k\lp s\rp$.  

We conclude this section with two results regarding fixed-point quantitative
theories.  The theorem below encodes, in terms of quantitative equational logic,
the fact that in a Banach theory (we will see later that they are interpreted in
1-bounded complete metric spaces) the sequence of iterations of a function $f$
on its i-th contractive variable, where the function is contractive, is a Cauchy
sequence that has as limit $\mu i.f$. Moreover, and here is the novelty that quantitative setting provides, we can monitor "the speed" of the convergence, and this provides us a powerful tool for building approximation theories.  

\begin{theorem}[Banach]\label{Banach}
	Let $\U$ be a Banach theory over $\Om X$ and $f:n:\theta\triangleright
    i\in\Om^\mu X$. Let $a=\max\{\alpha_i\mid\ol\alpha\in\theta\}$. We focus on
    the i-th variable of $f$, denoted $f\lp x_i\rp$. Then, 
	\\(1). $y=_\e z\vdash [f]_i^k\lp y\rp=_{\e a^k} [f]_i^k\lp z\rp\in\U^\mu$;
	\\(2). $y=_\e f\lp y\rp\vdash [f]_i^k\lp y\rp=_{\e a^k\frac{1-a^l}{1-a}}
    [f]_i^{k+l}\lp y\rp\in\U^\mu$; 
	\\ and for any $s\in\Om^\mu X$ and any $\ol t\in(\Om^\mu X)^n$,
	\\(3). $\forall\e>0~\exists k~\forall m~~\vdash [f]_i^k(\ol
    t[s/i])=_\e[f]_i^{k+m}(\ol t[s/i])\in\U^\mu$; 
	\\(4). $\forall\e>0~\exists k~\forall m~~\vdash [f]_i^{k+m}(\ol
    t[s/i])=_\e\mu i.f(\ol t\setminus i)\in\U^\mu$. 
\end{theorem}

\begin{proof}
	A consequence of the (Banach) rule is that \\$y=_\e z\vdash f\lp y\rp=_{\e a} f\lp
    z\rp\in\U^\mu$.  We apply this repeatedly to get (1) and use (Triang) to get
    (2). 
	\\ To prove (3), we start from $\vdash y=_1 f\lp y\rp\in\U^\mu$ which we get
    from ($1$-bound) and apply (2) observing that since $a<1$,
    $a^k\frac{1-a^l}{1-a}$ can be made arbitrarly small for any
    $l$ by choosing a sufficiently large $k$.
	\\ For (4) we first observe that from (3) we get that 
	$$\forall\e>0,~\exists k~\forall m,~\vdash[f]_i^{k+m}(\ol t[s/i])=_{\e(1-a)}[f]_i^{k+m+1}(\ol
    t[s/i])\in\U^\mu.$$  We use this in (Approx) instantiated with
    $u=[f]_i^{k+m}(\ol t[s/i])$ and $\e=\frac{\e}{1-a}$ and we get (4). 
\end{proof}

We can talk about convergent sequences in $\Om^\mu X$, in the same way that we
discussed them in quantitative algebras.  Note that the previous theorem
provides an important limit argument: namely the fixed point is obtained as the
limit of the iterates.  This is, of course, how the Banach fixed-point theorem
is supposed to work.  These results show how Banach-style reasoning is
internalized in quantitative logic.

\begin{corollary}\label{mu-limit}
  Let $\U$ be a Banach theory over $\Om X$ and
  $f:n:\theta\triangleright i\in\Om^\mu X$. Then for any $\ol t\in(\Om^\mu X)^n$
  and any $s\in\Om^\mu X$, $([f]_i^k(\ol t[s/i]))_{k\geq 1}$ is a convergent
  sequence in $\U^\mu$ and moreover,
  $$\vdash\lim_k [f]_i^k(\ol t[s/i])=_0\mu i.f(\ol t\setminus i)\in\U^\mu.$$
\end{corollary}

The next theorem shows that $\mu i.f$ is indeed the unique parametric fixed
point of $f$ in its $i$-th variable.

\begin{theorem}[Parametric fixed-point]\label{fixed-point}
	Let $\U$ be a Banach theory over $\Om X$ and $f:n:\theta\triangleright
    i\in\Om^\mu X$. Then, for any $s\in\Om^\mu X$ and any $\ol t\in(\Om^\mu
    X)^n$, 
	\\(1). $\vdash\mu i.f(\ol t\setminus i)=_0 f(\ol t[\mu i.f(\ol t\setminus
    i)/i])\in\U^\mu$; 
	\\(2). $s=_0 f(\ol t[s/i])\vdash s=_0 \mu i.f(\ol t\setminus i)\in\U^\mu.$
\end{theorem}

\begin{proof}
	 Let $a=\max\{\alpha_i\mid\ol\alpha\in\theta\}$. 
	\\(1). From Theorem \ref{Banach} (4), $\forall\e>0\;\exists k\;\forall m$, 
	\\$\vdash[f]_i^{k+m-1}(\ol t[s/i])=_{\frac{\e}{2a}}\mu i.f(\ol t\setminus
    i)\in\U^\mu$. And applying Theorem \ref{Banach} (1) to it we get
    \\$\vdash[f]_i^{k+m}(\ol t[s/i])=_{\frac{\e}{2}}f(\ol t[\mu i.f(\ol
    t\setminus i)/i])\in\U^\mu$. 
	\\ On the other hand, Theoren \ref{Banach} (4) also guarantees that
    \\$\vdash[f]_i^{k+m}(\ol t[s/i])=_{\frac{\e}{2}}\mu i.f(\ol t\setminus
    i)\in\U^\mu$. 
	\\We apply (Triang) to the previous two equations and get that for any
    $\e>0$, \\$\vdash\mu i.f(\ol t\setminus i)=_\e f(\ol t[\mu i.f(\ol
    t\setminus i)/i])\in\U^\mu$. Now (Cont) concludes the proof. 
	\\(2). Now we instantiate (Approx) with $u=s$ and $\e=0$.
\end{proof}


\subsection{Semantics of fixed-point judgements}

The fixed-point quantitative theories will be interpreted on quantitative
algebras over $1$-bounded complete metric spaces.  

Let $\Omega$ be a Banach signature and $\CM(\Omega)$ the category of
$\Omega$-quantitative algebras over 1-bounded complete metric spaces. If $\U$ is
a quantitative equational theory over $\Om X$, let $\CM(\U)$ denote the class of
models of $\U$ in $\CM(\Omega)$. 

Let $\A=(A,\Omega,d)\in\CM(\Omega)$, $f:A^n\to A$, $a\in A$ and $i\leq n$. We
define the sequence of iterations of $f$ on $a$ for its $i$-th variable, which is
the family of functions $[f]_i^k:A^{n-1}\to A$ for $k\in\mathbb N$, inductively
as follows, where $\ol x=\ang{x_1..x_n}$ is a sequence of variables 
$$[f]_i^1(\ol x\setminus i)=f(\ol x[a/i]),$$
$$[f]_i^{k+1}(\ol x\setminus i)=f(\ol x[[f]_i^k(\ol x\setminus i)/i]).$$

We know from Banach's fixed-point theorem~\cite{Banach22} that if $f$ is
contractive in its $i$-th variable, then the sequence $([f]_i^k)$ is Cauchy and
has a unique limit, which can be achieved by iterating $f$ on any element of
$A$.  Let $f^*_i:A^{n-1}\to A$ denote this limit; this is a function of the
remaining $n-1$ paramemeters and gives the fixed point in the iterated
position.  

We will use this fact to interpret any fixed-point term in $\A$. Suppose that $t:n:\theta\triangleright i\in\Om^\mu X$. Then, $t^\A:A^n\to A$ is $\max\{\alpha_i\mid\ol\alpha\in\theta\}$-contractive in its i-th variable. Hence, applying Banach Theorem we have that there exists $$[t^\A]^*_i:A_{n-1}\to A.$$
We use this to interpret $\mu i.t$ in $\A$ by defining
$$(\mu i.t)^\A=[t^\A]^*_i.$$
In this way, all the terms in $\Om^\mu X$ can be interpreted in $\A$. And this
allows us to interpret any quantitative equation and any quantitative judgement
by canonically extending the usual definition as follows.

Given an assignment $\iota\in\Omega[X|\A]$ and a tuple $\ol z=\ang{z_1..z_n}\in
X^n$, let $\iota(\ol z)=\ang{\iota(z_1)..\iota(z_n)}$. With this notation, we
extend $\iota$ canonically, from $\Om X$ to $\Om^\mu X$, by letting for any $\mu
i.t(\ol z\setminus i)\in \Om^\mu X$, $$\iota(\mu i.t(\ol z\setminus
i))=[t^\A]_i^*(\ol{\iota(z\setminus i)}).$$ 

\begin{df}[Satisfaction for fixed-point judgements]
  Let $\Omega$ be a Banach signature and $\A\in\CM(\Omega)$.  Let
  $$\{s_i=_{\e_i}t_i\mid i\in I\}\vdash s=_\e t\in\mathcal J(\Om^\mu X).$$
  Then, for $\iota\in\Omega[X|\A]$, we write
  $$\{s_i=_{\e_i}t_i\mid i\in I\}\models_{\A,\iota} s=_\e t,$$ if
	$[\forall i\in I,~d^\A(\iota(t_i),\iota(s_i))\leq \e_i]$
	implies $d^\A(\iota(s),\iota(t))\leq\e$.	
	\\Similarly, for any $\Gamma\vdash\phi\in\mathcal J(\Om^\mu X)$,
	$$\Gamma\models_\A\phi\text{ iff }\forall\iota\in\Omega[X|\A],~\Gamma\models_{\A,\iota}\phi$$
	and for a set $\M\subseteq\CM(\Omega)$,
	$$\Gamma\models_\M\phi\text{ iff }\forall\A\in\M,~\Gamma\models_{\A}\phi.$$
\end{df}

The next theorem states that for a Banach theory the set of its models coincides
with the set of models of its fixed-point extension. 
\begin{theorem}\label{t1}
	Let $\U$ be a Banach theory over $\Om X$ and $\U^\mu$ its fixed-point
    extension. Then, for any $\A\in\CM(\Omega)$,
    $$\A\models\U\text{ iff }\A\models\U^\mu.$$ 
\end{theorem}
\begin{proof}
  The right-to-left implication follows from $\U\subseteq\U^\mu$.  We prove the
  left-to-right implication as follows. It is sufficient to demonstrate that any
  $\A\in \CM(\Omega)$ satisfies (Approx).\\
  Let $f:n:\theta\triangleright i\in\Om^\mu X$ and
  $a=\max\{\alpha_i\mid\ol\alpha\in\theta\}$.  We need to prove that for any
  $\ol t\in(\Om^\mu X)^n$, any $s\in\Om^\mu X$ and any $\e\geq 0$,
  $s=_\e f(\ol t[s/i])\models_\A s=_{\frac{\e}{1-a}}\mu i.f(\ol t\setminus i)$.
  \\Let $\iota\in\Omega[X|\A]$ and let $\sigma=\iota(s)$ and
  $\ol\tau=\iota(\ol t)$.  \\Suppose that
  $d(\sigma,f(\ol\tau[\sigma/i]))\leq\e$. Let $m=\mu i.f(\ol\tau\setminus i)$.
  Then, $m=f(\ol\tau[m/i])$. We have \\$d(\sigma,m)\leq
  d(\sigma,f(\ol\tau[\sigma/i]))+d(f(\ol\tau[\sigma/i]),m)$
  \\$=d(\sigma,f(\ol\tau[\sigma/i]))+d(f(\ol\tau[\sigma/i]),f(\ol\tau[m/i]))$
  \\$\leq \e+ a d(\sigma,m)$. Hence,
  $d(\sigma,m)\leq\frac{\e}{1-a}$ implying that
  \\$s=_\e f(\ol t[s/i])\models_{\A,\iota} s=_{\frac{\e}{1-a}}\mu i.f(\ol
  t\setminus i)$.  This concludes the proof.
\end{proof}

Hence, the class of models of $\U$ and of $\U^\mu$ coincide in the category of
$1$-bounded complete metric spaces.  

\begin{remark}\label{signature}
  Note that all the terms in $\Om^\mu X$ are nonexpansive in all their variables
  and with well-defined Banach patterns. Consequently, we can think of theses
  terms as elements of a larger Banach signature $\Om^\mu$ that contains all the
  terms as function symbols, and this is a "legal" quantitative algebra
  signature. Similarly, one can think of the fixed-point extension $\U^\mu$ of a
  quantitative equational theory $\U$ over a Banach signature as a 
  quantitative equational theory over the signature $\Om^\mu X$ as originally
  defined in \cite{Mardare16}.
\end{remark}

This remark together with the result of Theorem~\ref{t1} allows us to conclude
this section with a completeness result. 
\begin{theorem}[Completeness for fixed-point theories]\label{completeness'}
	Let $\U$ be a Banach theory over $\Om X$ and $\U^\mu$ its fixed-point extension. Then for any fixed-point quantitative judgement $\Gamma\vdash\phi\in\mathcal J(\Om^\mu X)$,
	$$\Gamma\models_{\CM(\U)}\phi\text{ iff }\Gamma\vdash\phi\in\U^\mu.$$
\end{theorem}

\begin{proof}
  Following the Remark \ref{signature}, $\Om^\mu$ is a Banach signature and
  $\U^\mu$ is a quantitative equational theory over $\Om^\mu X$. From the
  completeness result for quantitative algebras, stated in~\ref{completenessQA}
  and proven in \cite{Mardare16}, we get that
		$$\Gamma\models_{\CM(\U^\mu)}\phi\text{ iff }\Gamma\vdash\phi\in\U^\mu.$$
        Applying Theorem~\ref{t1}, which asserts that
        $\CM(\U)=\CM(\U^\mu)$, we conclude the proof.
\end{proof}


\section{Quantitative Fixed-Point Theories}

In this section we investigate the relation between the fixed-point extension of
Banach theories and the traditional concepts of Conway theories and iteration
theories~\cite{Bloom93,Simpson00}. 

\textbf{Notation.} In what follows, for a term $f:n:\theta\triangleright i$
and a sequence $\ol x=\ang{x_1..x_n}$ of variables, we will also use the usual
variable-binding fixed-point syntax and write $\mu x_i.f(\ol x)$ to denote $\mu
i.f(\ol x\setminus i)$.  This notation allows us to present a series of 
results in a more familiar format.

\subsection{Quantitative Conway theories}

The Conway theories \cite{Bloom93,Simpson00}, are defined by two properties Dinaturality and Diagonal property. We prove here that quantitative versions of these can be proven in any fixed-point Banach theory.

\begin{lemma}\label{l5}
	Let $\U$ be a Banach theory over $\Om X$ and let $f:p:\theta\triangleright
    i,~g:q:\zeta\triangleright j\in\Om^\mu X$. We focus on the $i$-th variable of
    $f$, $f\lp x_i\rp$ and on the $j$-th variable of $g$, $g\lp y_j\rp$.   
	\\Then, for any $s\in\Om^\mu X$ we have that $\forall\e>0~\exists n~\forall m$, 
	$$\vdash [f\lp g\lp y_j\rp\rp]_{i+j-1}^{n+m}(s)=_\e f\lp[g\lp f\lp
    x_i\rp\rp]_{i+j-1}^{n+m}(s)\rp\in\U^\mu.$$ 
\end{lemma}

\begin{proof}
  Let $a=\max\{\alpha_i\mid\ol\alpha\in\theta\}$,
  $b=\max\{\beta_j\mid\ol\beta\in\zeta\}$.  From ($1$-bound) we have
  $\vdash s=_1 f\lp s\rp\in\U^\mu$ and applying (Banach) to this, we conclude
  $\vdash g\lp s\rp=_b g\lp f\lp s\rp\rp\in\U^\mu$.  We again apply (Banach) to
  this last equation and get
  $\vdash f\lp g\lp s\rp\rp=_{ab}f\lp g\lp f\lp s\rp\rp\rp\in\U^\mu$. Hence,
  \\$\vdash [f\lp g\lp y_j\rp\rp]_{i+j-1}^1(s)=_{ab} f\lp[g\lp f\lp
  x_i\rp\rp]_{i+j-1}^1(s)\rp\in\U^\mu.$  Repeating these steps we obtain
  \\$\vdash [f\lp g\lp y_j\rp\rp]_{i+j-1}^n(s)=_{(ab)^n} f\lp[g\lp f\lp
  x_i\rp\rp]_{i+j-1}^n(s)\rp\in\U^\mu.$ \\Since $ab<1$, we can make $(ab)^n$ as
  small as we want.
\end{proof}

This lemma allows us to prove a quantitative version of the Dinaturality
property~\cite{Simpson00}. 
\begin{theorem}[Quantitative Dinaturality]\label{dinaturality}
	Let $\U$ be a Banach theory over $\Om X$ and $f:n:\theta$,
    $g:m:\zeta\in\Om^\mu X$ such that $\theta[\zeta/i]$ and $\zeta[\theta/j]$
    are $i+j-1$-contractive.  We focus on the i-th variable of $f$, $f\lp
    x_i\rp$, and on j-th variable of $g$, $g\lp y_j\rp$. 
	Then, $$\vdash\mu y_j.f\lp g\lp y_j\rp\rp=_0 f\lp \mu x_i. g\lp f\lp x_i\rp\rp\rp\in\U^\mu.$$
\end{theorem}

\begin{proof}
	Let $a=\max\{\alpha_i\mid\ol\alpha\in\theta\}$.
	\\	By using Lemma~\ref{l5} and Theorem~\ref{Banach} (4) together, we obtain
    $\forall \e>0~\exists n~\forall m$ such that the following three statements
    are satisfied. 
	\\$\vdash[f\lp g\lp y_j\rp\rp]_{i+j-1}^{n+m}(s)=_{\frac{\e}{3}}f\lp[g\lp f\lp x_i\rp\rp]_{i+j-1}^{n+m}(s)\rp\in\U^\mu,$
	\\$\vdash[f\lp g\lp y_j\rp\rp]_{i+j-1}^{m+n}(s)=_{\frac{\e}{3}}\mu y_j. f\lp g\lp y_j\rp\rp\in\U^\mu$ and 
	\\$\vdash [g\lp f\lp x_i\rp\rp]_{i+j-1}^{m+n}(s)=_{\frac{\e}{3a}}\mu x_i.g\lp f\lp x_i\rp\rp\in\U^\mu$. 
	\\The last one implies 
	\\$\vdash f\lp[g\lp f\lp x_i\rp\rp]_{i+j-1}^{m+n}(s)\rp=_{\frac{\e}{3}}f\lp\mu x_i.g\lp f\lp x_i\rp\rp\rp\in\U^\mu$.
	\\Using this one and the first two with (Triang) we conclude that for any $\e>0$,
	$$\vdash\mu y_j.f\lp g\lp y_j\rp\rp=_\e f\lp \mu x_i. g\lp f\lp
    x_i\rp\rp\rp\in\U^\mu.$$ Now we apply (Cont) and complete the proof. 
\end{proof}
This type of ``$\epsilon/3$-argument'' is common in analysis.

With these results in hand we can proceed and prove a quantitative version of
the diagonal property for fixed-point theories.

\begin{theorem}[Quantitative Diagonal property]\label{diagonal}
  Let $\U$ be a Banach theory over $\Om X$ and $f:n:\theta\in \Om^\mu
  X$. Suppose there exists $i<j\leq n$ s.t. for any $\ol\alpha\in\theta$,
  $\alpha_i+\alpha_j<1$. We focus on the i-th and j-th variables of $f$,
  $f\lp x_i,x_j\rp$. Then,
	$$\vdash\mu x_i.f\lp x_i,x_i\rp=_0 \mu x_j.\mu x_i.f\lp x_i,x_j\rp\in\U^\mu.$$
\end{theorem}

\begin{proof}
	Let $s=\mu x_i.f\lp x_i,x_i\rp$ and $t\lp x_i\rp=\mu x_j.f\lp x_i,x_j\rp$. Theorem \ref{fixed-point} (1) guarantees that \\$\vdash s=_0 t\lp s,s\rp\in\U^\mu$ and 
	$\vdash t\lp x_i\rp=_0 f\lp x_i,t\lp x_i\rp\rp\in\U^\mu.$
	\\Let $a=\max\{\alpha_i,\alpha_j\mid\ol\alpha\in\theta\}$. Then applying (Banach),
	\\$x_i=_\e s\vdash f\lp x_i,t\lp x_i\rp\rp=_{a\e}f\lp s,t\lp x_i\rp\rp\in\U^\mu$ and
	\\$t\lp x_i\rp=_\delta s\vdash f\lp s,t\lp x_i\rp\rp=_{a\delta} f\lp s,s\rp\in\U^\mu$. We apply (Triang) and obtain
	\\$\{x_i=_\e s,~t\lp x_i\rp=_\delta s\}\vdash f\lp x_i,t\lp x_i\rp\rp =_{a(\e+\delta)}f\lp s,s\rp\in\U^\mu$. Now if we instantiate this with $x_i=s$, we get
	\\$t\lp s\rp=_\delta s\vdash f\lp s,t\lp s\rp\rp =_{a\delta}f\lp s,s\rp\in\U^\mu$.
	We already know that \\$\vdash s=_0 t\lp s,s\rp\in\U^\mu$ and 
	$\vdash t\lp x_i\rp=_0 f\lp x_i,t\lp x_i\rp\rp\in\U^\mu.$ Combining these three, we obtain
	\\$t\lp s\rp=_\delta s\vdash t\lp s\rp=_{a\delta} s\in\U^\mu$. By applying
    this repeatedly and eventually using (Cont), we obtain $$\vdash t\lp s\rp=_0
    s\in\U^\mu,$$  which is the desired result.
\end{proof}

Theorems \ref{dinaturality} and \ref{diagonal} establish that the fixed-point
extension of any Banach theory is a Conway theory, in the sense of
\cite{Simpson00}.  


\subsection{Quantitative Iteration Theories}

In this subsection we show that the fixed-point Banach theories are not only
Conway theories, but they are iteration theories in the sense of
\cite{Simpson00}; meaning that, in addition to quantitative dinaturality and the
quantitative diagonal property, they also satisfy a quantitative version of the
amalgamation property.

\begin{lemma}\label{ovelap-mu}
	Let $\U$ be a Banach theory over $\Om X$ and $f:n:\theta,
    g:m:\zeta\in\Om^\mu X$ such that there exist $i<j\leq n$ and $u<v\leq m$
    with $a_i(1-b_v)+b_u<1$ and $a_i+2a_jb_v<1$, where
    $a_i=\max\{\alpha_i\mid\ol\alpha\in\theta\}$,
    $a_j=\max\{\alpha_j\mid\ol\alpha\in\theta\}$,
    $b_u=\max\{\beta_u\mid\ol\beta\in\zeta\}$ and
    $b_v=\max\{\beta_v\mid\ol\beta\in\zeta\}$.  We focus on the variables $i$ and $j$ in $f$, $f\lp x_i,x_j\rp$ and on variables $u$ and $v$ in $g$, $g\lp y_u,y_v\rp$. Then,
	$$\vdash\mu x.f\lp x, \mu y. g\lp x,y\rp\rp=_0 \mu x.\mu y. f\lp x, g\lp x,y\rp\rp\in\U^\mu.$$
\end{lemma}

\begin{proof}
	The inequalities $a_i(1-b_v)+b_u<1$ and $a_i+2a_jb_v<1$ guarantee that the
    fixed-points are properly defined. 
	
	Observe now that by repeatedly applying~\ref{limit1} (2) we can prove that
    for any two sequences $(s_k)_{k\geq 1}\subseteq\Om^\mu X$ and
    $(t_r)_{r\geq 1}\subseteq\Om^\mu X$ convergent in $\displaystyle\U^\mu$,
    $(\lim_{r}f\lp s_k, g\lp s_k,t_r\rp\rp)_{k\geq 1}$ and
    $\displaystyle(f\lp s_k, \lim_{r} g\lp s_k,t_r\rp\rp)_{k\geq 1}$ are
    convergent in $\U^\mu$ and moreover,
	$$\vdash \lim_k f\lp s_k, \lim_{r} g\lp s_k,t_r\rp\rp =_0
    \lim_k(\lim_{r}f\lp s_k, g\lp s_k,t_r\rp\rp) \in\U^\mu.$$ 
	Applying this in the context of Corollary \ref{mu-limit}, we get the desired
    result.
\end{proof}

With the result of the previous lemma, we are ready to state the
quantitative amalgamation theorem. 

\begin{theorem}[Quantitative amalgamation]\label{amalgamation}
	Let $\U$ be a Banach theory over $\Om X$ and let $f_i:n:\theta^i\in\Om^\mu X$
    for $i\leq n$ be a family of functions such that for any $i\neq j$ and any
    $\ol\alpha^i\in\theta^i$ and $\ol\alpha^j\in\theta^j$, $$\sum_{k\leq
      n}\alpha^i_k=\sum_{k\leq n}\alpha^j_k=\alpha<1.$$ 
	Suppose there exists $g:1:\{\ang{\alpha}\}\in\Om^\mu X$ s.t. for all $i\leq n$,
	$$\vdash f_i(x..x)=_0 g(x)\in\U^\mu.$$
	If there exists $s_1,..,s_n\in\Om^\mu X$ s.t. for all $i\leq n$,
	$$\vdash s_i=_0 f_i(s_1..s_n)\in\U^\mu,$$
	then for all $i\leq n$, $$\vdash s_i=_0\mu x.g(x)\in\U^\mu.$$ 
\end{theorem}

\begin{proof}
	We only sketch the proof for the case $n=2$ that is simpler to present. The
    general case is proven in the same way, but one needs to keep
    track of more indices. 
	
	From the hypothesis, \\$\vdash s_1=_0 f(s_1,s_2)\in\U^\mu$ and $\vdash
    s_2=f_2(s_1,s_2)\in\U^\mu$.  Let $t$ denote $\mu
    x.g(x)$. From Theorem~\ref{fixed-point} we know that $\vdash t=_0
    g(t)\in\U^\mu$. \\ We will prove that for any $i$, $\vdash s_i=_0
    t\in\U^\mu$.
	
	Let $\phi_1(x_2)=\mu x_1.f_1(x_1,x_2)$ and $\phi_2(x_1)=\mu x_2.f_2(x_1,x_2)$. From the hypothesis we have $\vdash \phi_1(s_2)=_0 s_1\in\U^\mu$ and $\vdash\phi_2(s_1)=_0 s_2\in\U^\mu$. Consequently,
	$\vdash s_1=\phi_1(\phi_2(s_1))\in\U^\mu$. Applying Theorem \ref{fixed-point}, we get then 
	\\$\vdash s_1=_0 \mu z.\phi_1(\phi_2(z))\in\U^\mu$. By extending $\phi_2$ we get further
	$\vdash s_1=_0 \mu z.\phi_1(\mu x_2.f_2(z,x_2))\in\U^\mu$, and after extending $\phi_1$ we get
	\\	$\vdash s_1=_0 \mu z.\mu x_1.f_1(x_1,\mu x_2.f_2(z,x_2))\in\U^\mu$. We
    use the quantitative diagonal property and get 
	\\$\vdash s_1=_0 \mu v.f_1(v,\mu x_2.f_2(v,x_2))\in\U^\mu$. Now we apply
    Lemma~\ref{ovelap-mu} to obtain 
	\\$\vdash s_1=_0 \mu v.\mu x_2.f_1(v,f_2(v,x_2))\in\U^\mu$.
	\\Lemma~\ref{ovelap-mu} also gives us\footnote{The same sequence of operations can be used $n$ times if the arity of $f$ is $n$ and get a fixed-point as the one we get here.}
	\\$\vdash s_1=_0 \mu w.f_1(w,f_2(w,w))\in\U^\mu$.
	\\From the hypothesis we know that $\vdash t=_0f_1(t,t)\in\U^\mu$ and
    $\vdash t=_0f_2(t,t)\in\U^\mu$. Hence, $\vdash t=_0
    f_1(t,f_2(t,t))\in\U^\mu$. Next Theorem \ref{fixed-point} guarantees that  
	\\$\vdash t=_0 \mu w. f_1(w,f_2(w,w))\in\U^\mu$.  Combining this with the
    previous fixed-point description that we derived for $s_1$, we get $\vdash
    s_1=_0 t\in\U^\mu$.  Similarly one can prove $\vdash s_2=_0 t\in\U^\mu.$ 
\end{proof}

Note that all three theorems have a top-level statement that is stated in terms
of exact equality but the proofs use approximate equality.  In the amalgamation
proof the approximate reasoning is isolated into Lemma~\ref{ovelap-mu}.  In
addition to Theorems~\ref{dinaturality} and~\ref{diagonal}, Theorem
\ref{amalgamation} guarantees that any fixed-point extension of a Banach theory
is an iteration theory as defined in \cite{Simpson00}.


\section{The metric coinduction principle}

In this section we will investigate the relation of these theories with a very
interesting and useful coinduction principle proposed by Kozen
in~\cite{Kozen06,Kozen07}.  We will demonstrate that the metric coinduction
principle can be proven within any Banach fixed-point theorem and that this
principle is equivalent to our rule (Approx).  That being said, however, if we
restrict ourselves to \emph{finitary} proofs, we suspect that the metric coinduction principle
is more powerful.

The context in which the metric coinduction principle is stated
in~\cite{Kozen06,Kozen07} is a bit more liberal than the syntax of fixed-point
Banach theories, as it involves the concept of closed predicate, defined as a
predicate whose extension is a closed set in any bounded complete metric space.
For this reason, we will work in this section at a
metalevel, where semantics concepts, \emph{i.e.}\ metric and
topological concepts, are used together with the syntax of Banach theories.

Consider a Banach theory $\U$ over $\Om X$ and its fixed-point extension
$\U^\mu$.  A \emph{closed predicate} in this context is any predicate $P$ whose
extension, when interpreted in any model in $\CM(\U)$, is a closed set in the
open-ball topology induced by the metric.  In this setting, the metric
coinduction principle for the closed predicate $P$ is stated as follows, for any map
$f:n:\theta\triangleright i\in\Om^\mu X$, any $\ol y\in X^n$ and an arbitrary
$t\in\Om^\mu X$.
\begin{align*}
	\text{\textbf{(MCoind)}} \quad 
	& \frac{\vdash P(t)~~P(x)\vdash P(f(\ol y[x/i]))}{\vdash P(\mu i.f(\ol y))}\,.
\end{align*} 

Given a Banach theory $\U$ over $\Om X$, let $\U^M$ be the smallest extension of
$\U$ over $\Om^\mu X$ that is closed under the metric coinduction principle
(MCoind) - we call it the coinductive extension of $\U$. 

The next two theorems will relate $\U^M$ and $\U^\mu$.

\begin{theorem}\label{coind1}
	Let $\U$ be a Banach theory over $\Om X$, and let $\U^M$ and $\U^\mu$ be its
    coinductive extension and fixed-point extension
    respectively. Then $$\U^\mu\subseteq\U^M.$$ 
\end{theorem}

\begin{proof}
	To prove this result it is sufficient to show that $\U^M$ is closed under the rule (Approx).
	
	For simplicity, we focus on the $i$-th variable of $f$, $f\lp x_i\rp$ and let
    $a=\max\{\alpha_i\mid\ol\alpha\in\theta\}$. 
	
	Consider the predicate $$R(y)=\forall x.(x=_\e f\lp x\rp\vdash
    x=_{\frac{\e}{1-a}}y\in\U^M).$$ and let $\ol{B_\e(x)}=\{z\in\Om^\mu X\mid~
    \vdash x=_\e z\}$, which is interpreted in any model as the $\e$-closed ball
    centred at $x$. 
	\\	Then we can characterize $R$ as follows
	$$R(y)=\forall x(f\lp x\rp\in\ol{B_\e(x)}\Rightarrow y\in\ol{B_{\frac{\e}{1-a}}(x)}~)$$
	$$=y\in\bigcap_{z\in\{x\mid~f\lp
      x\rp\in\ol{B_\e(x)}\}}\ol{B_{\frac{\e}{1-a}(x)}}.$$ Hence, $R$ is a closed
    predicate and we can use it to instantiate (MCoind) and conclude that $\U^M$ is
    closed under the rule  
	$$\frac{\vdash R(t)~~R(x)\vdash R(f\lp x\rp)}{\vdash R(\mu x.f\lp x\rp)}.$$
	We prove now that for any $x$, $R(x)\vdash R(f\lp x\rp)\in\U^M$ which is equivalent to proving that 
	$$\forall z[\forall x(x=_\e f\lp x\rp\vdash x=_{\frac{\e}{1-a}}z)\in\U^M$$
	$$\Rightarrow\forall x(x=_\e f\lp x\rp\vdash x=_{\frac{\e}{1-a}}f\lp x\rp)\in\U^M].$$
	Suppose that for any $x$, $x=_\e f\lp x\rp\vdash x=_{\frac{\e}{1-a}} z\in\U^M$. Since $f$ is contractive, (Banach) guarantees that \\$x=_\e y\vdash f\lp x\rp=_{a\e}f\lp y\rp\in\U^M$. Hence,
	\\$x=_\e f\lp x\rp\vdash f\lp x\rp=_{\e\frac{a}{1-a}} f\lp z\rp\in\U^M$. Next (Triang) proofs
	\\$x=_\e f\lp x\rp\vdash x=_{\frac{\e}{1-a}} f\lp z\rp\in\U^M$, hence for
    any $x$, $R(x)\vdash R(f\lp x\rp)\in\U^M$. 
	
	Now it is not difficult to notice that $x=_\e f\lp x\rp\vdash
    x=_{\e\frac{\-a^n}{1-a}}[f]^n_i\lp s\rp$ - Theorem \ref{Banach}. So, since
    ($1$-bound) guarantees that for any $s\in\Om^\mu X$, $\vdash s=_1 f\lp
    s\rp\in\U^M$, we get that the sequence $([f]_i^k\lp s\rp)_{k\geq 1}$ is
    convergent in $\U^M$ and its limit $t$ is such that $\vdash R(t)\in\U^M$.  
	
	Hence both hypothesis of (MCoind) for $R$ are satisfied, meaning that its
    conclusion has to be true, which is $$\vdash R(\mu i.f)\in\U^M,$$ but this
    is exactly (Approx).   
\end{proof}

The next theorem says that whenever we have a closed predicate,
any consequences proved using (Mcoind) with this predicate
can be proved in $\U^{\mu}$.

\begin{theorem}\label{coind2}
  Let $\U$ be a Banach theory over $\Om X$, let $\U^{\mu}$ be its fixed-point
  extension and let $P$ be a closed predicate. Then any consequences of $P$
  obtained using (Mcoind) can be established in $\U^{\mu}$.

\end{theorem}

\begin{proof}
	Let $P$ be a closed predicate.  Then it must be the complement of an open predicate $B$, i.e., $$P=B^c.$$
	Let $$B_\e(x)=\{y\in\Om^\mu X\mid~\vdash x=_\delta y\in\U^\mu\text{ for some
    }\delta<\e\},$$ 
	be the $x$-centred open ball of radius $\e>0$.  These sets for a base in the
    open ball topology, hence there must exist a  set $I$ of indices and a set
    of $I$-indexed terms $s_i\in\Om^\mu X$ such that $$B=\bigcup_{i\in
      I}B_{\e_i}(s_i).$$ Consequently, $$P=\bigcap_{i\in I}B^c_{\e_i}(s_i).$$
    Now we have 
	$$B^c_\e(x)=\{y\mid\vdash x=_\delta y\in\U^\mu\Rightarrow\delta\geq\e\}.$$ Hence,
	$$P=\{y\mid\forall i\in I, \vdash s_i=_\delta
    y\in\U^\mu\Rightarrow\delta\geq\e_i\}.$$ Hence we can define any closed
    predicate $P$ as $$P(x)=\forall i\in I (\vdash s_i=_\delta
    x\in\U^\mu\Rightarrow\delta\geq\e_i).$$  Now we prove that $\U^\mu$ is closed
    under (MCoind) for $P$. 
	\\Suppose that for some $s\in\Om^\mu X$, $\vdash P(s)\in\U^\mu$, and that
    $P(x)\vdash P(f\lp x\rp)\in\U^\mu$. The second one means
	$$\forall x[\forall i(\vdash x=_\delta s_i\in\U^\mu\Rightarrow\delta\geq\e_i)$$
	$$\Rightarrow\forall i(\vdash f\lp x=_\delta
    s_i\in\U^\mu\Rightarrow\delta\geq\e_i)].$$ Iterating this over $$\forall
    i\in I(\vdash s_i=_\delta s\in\U^\mu\Rightarrow\delta\geq \e_i),$$ which is
    an equivalent statement for $\vdash P(s)\in\U^\mu$, we get 
	$$\forall k\;\forall i\;[\vdash [f]_i^k\lp s\rp=_\delta s_i\Rightarrow\delta\geq
    \e_i] (*).$$
	We need to prove that $$\forall i[\vdash s_i=_\delta \mu x.f\lp
    x\rp\Rightarrow\delta\geq\e_i].$$ 
	Suppose this is not the case and there exists some $j\in I$ so that for some $r>0$, 
	$$\vdash\mu x.f\lp x\rp=_rs_j\in\U^\mu\;\wedge\; r <\e_j.$$
	We know from Corollary \ref{mu-limit} that for any $0<p<\e_j-r$ there exists
    some $k$ s.t. 
	$$\vdash\mu x.f\lp x\rp=_{\e_j-r-p}[f]_i^k\lp s_j\rp\in\U^\mu.$$  Finally
    (Triang) gives us  $$\vdash [f]_i^k\lp s_j\rp=_{\e_j-p}s_j\in\U^\mu,$$  but
    this contradicts the statement $(*)$ above since $\e_j-p <\e_j$. 
\end{proof}

The results stated in Theorems \ref{coind1} and \ref{coind2} show that
the metric coinduction principle, despite its more semantic flavour and its
quantification over all closed predicates, has the same power as our fixed
point Banach theories.  However, it is often easier to use and is a very
attractive proof principle.


\section{Markov Decision Processes and \\ the Bellman equation} 
Markov decision processes~\cite{Puterman94} are a well known formalism used in
operations research and extensively in reinforcement learning~\cite{Sutton98}.
The Bellman equation is perhaps the most common application of the Banach
fixed-point theorem.  This section is an extended example showing how one can
reason about the Bellman equation in our setting.  Indeed this research project
began from a desire to treat the Bellman equation as an example within the
quantitative equational logic framework before we developed the general theory
reported here.


\subsection{Markov decision processes}
\begin{df}
A Markov decision process is a tuple $$\MM=(S,A,(P^a)_{a\in A},(R^a)_{a\in A})$$ where 
\begin{itemize}
	\item $S$ is a finite set of \emph{states}; let $\Delta S$ represent the set
      of probability distributions on $S$. 
	\item $A$ is a finite set of \emph{actions}; let $\Delta A$ represent the
      set of probability distributions on $A$. 
	\item For each $a\in A$, $P^a:S\to\Delta(S)$ are the labelled probabilistic
      \emph{transitions}. 
	\item For each $a\in A$, $R^a: S\to[0,1]$ is the \emph{reward} function.
    \end{itemize}
\end{df}
One can think of these as transition systems where an external agent controls
the system choosing actions according to some policy.  The system responds by
changing state according to the transition function and returning a reward.  The
reward is accumulated, with a multiplicative discount factor, and the goal of
reinforcement learning is to find the best policy for optimizing the reward.

The effectiveness of a particular policy is captured by what are called
\emph{value functions} which summarize the aggregated discounted rewards
associated with a policy.  Mathematically, value functions are elements of the
space $\V=[0,1]^S$; this is a metric space endowed with the
metric $$d(f,g)=\max_{s\in S}|f(s)-g(s)|.$$ A \emph{policy} is a map
$\pi:S\to\Delta A$ that associates to each state a probability distribution over
the actions.  Let $\Pi$ denote the set of policies for $\MM$.
For arbitrary $a\in A$ we write $\hat a$ for the constant policy that associates
to any state the Dirac distribution concentrated at $a$.

For an arbitrary policy $\pi\in\Pi$, the expected immediate reward of $\pi$ is
the value function $R^\pi\in \V$ defined for arbitrary $s\in S$,
by $$R^\pi(s)=\sum_{a\in A}\pi(s)(a) R^a(s).$$

Given a policy $\pi\in\Pi$ and a discount factor $\gamma\in(0,1)$, the Bellman
operator of $\pi$ is the operator $T^\pi:\V\to\V$ defined for arbitrary $f\in\V$
and $s\in S$ as follows 
$$T^\pi(f)(s)=(1-\gamma) R^\pi(s)+\gamma\sum_{a\in A}\sum_{s'\in
  S}\pi(s)(a)P^a(s)(s')f(s')$$ 

The Bellman equation for the policy $\pi\in\Pi$ and discount factor
$\gamma\in(0,1)$ is the following fixed point equation over $\V$ $$X=T^\pi(X).$$
The discount factor makes this operator contractive and thus has a unique fixed
point: this is the value function of the policy $\pi$.


\newcommand{\brck}[1]{\llbracket #1\rrbracket}

\subsection{Reward Barycentric Algebra}
\textbf{Assumptions} For the rest of this section, we assume a fixed
Markov decision process $\MM=(S,A,(P^a)_{a\in A},(R^a)_{a\in A})$ and a fixed
discount factor $\gamma\in(0,1)$.  

We develop a particular Banach theory, designed for
solving the Bellman equation for $\MM$ and $\gamma$.  Its signature extends the barycentric signature and the theory
extends the quantitative barycentric theory developed in~\cite{Mardare16}. The
models of our theory will be called reward barycentric algebras (RBA), and will
be a specialised class of barycentric algebras, as defined in~\cite{Mardare16}, devised with additional algebraic structure.

\textbf{Signature.} Consider the Banach signature $\Sigma$ containing the following basic operators.
\begin{itemize}
	\item For each $\e\in[0,1]$, $+_\e:2:\{\ang{\e,1-\e}\}\in\Sigma$;
	\item For each $\pi\in\Pi$, $\ang{\pi}:1:\{\ang{1}\}\in\Sigma$;
	\item For each $\pi\in\Pi$, $|\pi|:1:\{\ang{\gamma}\}\in\Sigma$.
\end{itemize}

Consider now the Banach theory $\B$ over $\hat\Sigma X$ axiomatized by the following two sets of axioms

\textbf{Barycentric axioms:} \\for arbitrary $\e,\e'\in[0,1]$, $p,q\in\preals$, $x,x',y,y'\in X$
\begin{description}[leftmargin=*]
	\item[(B1)] $\vdash x+_1x'=_0 x$
	\item[(B2)] $\vdash x+_\e x=_0 x$
	\item[(SC)] $\vdash x+_\e x'=_0 x'+_{1-\e}x$
	\item[(SA)] $\vdash (x+_\e x')+_{\e'}y=_0 x+_{\e\e'}(x'+_{\frac{\e'-\e\e'}{1-\e\e'}}y)$ for $\e\e'<1$
	\item[(BA)] $\{x=_p x', y=_q y'\}\vdash x+_\e x'=_{\e p+(1-\e)q}y+_\e y'$
\end{description}

\textbf{Reward axioms:} \\for arbitrary $\pi,\pi'\in\Pi$, $\e\in[0,1]$ and $x,y\in X$
\begin{description}[leftmargin=*]
	\item[(R1)] $\vdash \ang{\e\pi+(1-\e)\pi'}x=_0 \ang{\pi}x+_\e\ang{\pi'}x$
	\item[(R2)] $\vdash |\e\pi+(1-\e)\pi'| x =_0 |\pi| x+_\e |\pi'|x$
	\item[(R3)] $x=_\e y\vdash |\pi|x=_{\gamma\e}|\pi|y$
\end{description}

\bigskip
\textbf{Algebra of value functions.} The space $(\V,d)$ of value
functions of $\MM$ is a $1$-bounded complete metric space and has a natural
$\sigma$-algebra of Borel sets. We interpret the basic functions
in $\Sigma$, for arbitrary $f,g\in \V$, $\pi\in\Pi$ and $s\in S$ as follows 
\begin{itemize}
	\item $(f+_\e g)^\V=\e f^\V+(1-\e)g^\V$
	\item $\displaystyle(\ang{\pi}f)^\V(s)=\sum_{a\in A}\pi(s)(a)\sum_{s'\in S}P^a(s)(s')f^\V(s')$
	\item $(|\pi|f)^\V=(1-\gamma) R^\pi+\gamma f^\V$
\end{itemize}

It is not difficult to verify that the functions have indeed the expected Banach patterns,
hence $\V$ with the previous interpretation is indeed an algebra of the right
form.  Consider now $\hat\Sigma^\mu X$ the fixed-point extension of
$\hat\Sigma X$.

For simplicity, in what follows we denote the interpretation of any $t\in\hat\Sigma^\mu X$ in $\V$ by $\brck t$. We can now prove that $\V$ satisfies indeed the axioms of $\B$.

\begin{theorem}
	The space $\V$ of value functions of $\MM$ is a model for $\B$, $\V\models\B.$
\end{theorem}

\begin{proof}
	The fact that the Barycentric axioms are satisfied by $\V$ is already proven
    in \cite{Mardare16}. We prove here the soundness of the reward axioms. 
	
	(R1): for any $t\in\hat\Sigma^\mu X$,
	\\$\brck{\ang{\e\pi+(1-\e)\pi'}t}(s)$
	\\$\displaystyle=\sum_{a\in A}(\e\pi(s)(a)+(1-\e)\pi'(s)(a))\sum_{s'\in S}P^a(s)(s')\brck t(s')$
	\\$\displaystyle=\e\sum_{a\in A}\pi(s)(a)\sum_{s'\in S}P^a(s)(s')\brck t(s')+$
	\\$\displaystyle +(1-\e)\sum_{a\in A}\pi'(s)(a)\sum_{s'\in S}P^a(s)(s')\brck t(s')$
	\\$=\displaystyle\e\brck{\ang{\pi}t}(s)+(1-\e)\brck{\ang{\pi'}t}(s)$
	\\$=\brck{\ang{\pi}t+_\e\ang{\pi'}t}(s).$
	
	(R2): for any $t\in\hat\Sigma^\mu X$,
	\\$\displaystyle\brck{|\e\pi+(1-\e)\pi'|t}(s)=$
	\\$\displaystyle=(1-\gamma)R^{\e\pi+(1-\e)\pi'}(s)+\gamma \brck t(s)$
	\\$\displaystyle=(1-\gamma)\sum_{a\in A}R^a(s)(\e\pi+(1-\e)\pi')(s)(a)+\gamma\brck t(s)$
	\\$=\displaystyle\e((1-\gamma)\sum_{a\in A}R^a(s)\pi(s)(a)+\gamma\brck t(s))+$
	\\$\displaystyle+(1-\e)((1-\gamma)\sum_{a\in A}R^a(s)\pi'(s)(a)+\gamma\brck t(s))$
	\\$=\brck{|\pi|t+_\e|\pi'|t}(s).$
	
	(R3): for any $t,t'\in\hat\Sigma^\mu X$,
	\\$|\brck{|\pi| t}(s)-\brck{|\pi| t'}(s)|$
	\\$=|(1-\gamma)R^\pi(s)+\gamma\brck t(s)-(1-\gamma)R^\pi(s)-\gamma \brck{t'}(s)|$
	\\$=\gamma |\brck t(s)-\brck{t'}(s)|.$
\end{proof}


\subsection{Solving the Bellman equation iteratively}
We define now, for any $\pi\in\Pi$ a derived operator $O^\pi$ inductively on the
structure of the policy $\pi$ as follows. 
\begin{itemize}
	\item For $a\in A$, $O^{\hat a}t=|\hat a|\ang{\hat a}t$.
	\item For $\pi,\pi'\in\Pi$ and $\e\in[0,1]$, $$O^{\e\pi+(1-\e)\pi'}t=O^\pi+_\e O^{\pi'}.$$
    \end{itemize}
    
    Since all the distributions with finite support can be represented as convex
    combinations of Dirac distributions, any policy can be represented by a term
    with appropriately nested $+_\e$ operators on top of constant policies.
    Hence the definition of $O^\pi$ is complete.

    The following theorem states that $O^\pi$ is the syntactic counterpart of
    the Bellman operator $T^\pi$.

\begin{theorem}\label{Bellman-op}
	For any $\pi\in\Pi$ and any $t\in\hat\Sigma^\mu X$, $$\brck{O^\pi t}=T^\pi\brck t.$$
\end{theorem} 

\begin{proof}
	We prove this inductively on the structure of $\pi\in\Pi$. Let $s\in S$.
	
	For $\pi=\hat a$, $a\in A$,
	\\$\brck{O^{\hat a}t}(s)=\brck{|\hat a\ang{\hat a}t}(s)$
	\\$=(1-\gamma)R^a(s)+\gamma\brck{\ang{\hat a}t}(s)$
	\\$=(1-\gamma) R^a(s)+\gamma\sum_{s'\in S}P^a(s)(s')\brck t(s')
	=T^{\hat a}\brck t(s).$
	\\For $\e\pi+(1-\e)\pi'$ under the inductive hypothesis for $O^\pi t$ and $O^{\pi'}t$. We have
	\\$\brck{O^{\e\pi+(1-\e)\pi'}t}(s)=\brck{O^\pi t+_\e O^{\pi'}t}(s)$
	\\$=\e\brck{O^\pi t}(s)+(1-\e)\brck{O^{\pi'}t}(s)$
	\\$=\e T^\pi \brck{t}(s)+(1-\e)T^{\pi'}\brck{t}(s)$
	\\$=\e\sum_{a\in A}\pi(s)(a)[(1-\gamma)R^a(s)+$
	\\$\gamma\sum_{s'\in S} P^a(s)(s')\brck t(s')] +$
	\\$+(1-\e)\sum_{a\in A}\pi'(s)(a)[(1-\gamma)R^a(s)+$
	\\$+\gamma\sum_{s'\in S} P^a(s)(s')\brck t(s')]$
	\\$=\sum_{a\in A}(\e\pi(s)+(1-\e)\pi'(s))(a)[(1-\gamma)R^a(s)+$
	\\$+\gamma\sum_{s'\in S}P^a(s)(s')\brck t(s')]=T^{\e\pi+(1-\e)\pi'}\brck t(s).$
\end{proof}

Next we verify that $O^\pi$ has Banach pattern $\{\ang{\gamma}\}$.
\begin{lemma}\label{Bellman-pattern}
	For any $\pi\in\Pi$, $$O^\pi:1:\{\ang{\gamma}\}\in \hat\Sigma^\mu.$$
\end{lemma}

\begin{proof}
	We prove, inductively on the structure of $\pi$, that \\$x=_\e y\vdash O^\pi x=_{\gamma\e} O^\pi y\in\B.$
\\For $\pi=\hat a$, (NExp) for $\ang{\hat a}$ gives us 
	\\$x=_\e y\vdash \ang{\hat a}x=_\e\ang{\hat a}y\in\B$ and instantiating (R2),
	\\$\ang{\hat a}x=_\e \ang{\hat a}y\vdash |\hat a|\ang{\hat a}x=_{\e\gamma}|\hat a|\ang{\hat a}y\in\B$,
	\\ hence, $x=_\e y\vdash |\hat a|\ang{\hat a}x=_{\e\gamma}|\hat a|\ang{\hat a}y\in\B$.
	\\For $\e\pi+(1-\e)\pi'$, consider the inductive hypotheses
	\\$x=_\e y\vdash O^\pi x=_{\e\gamma}O^{\pi}y\in\B$ and 
	\\$x=_\e y\vdash O^{\pi'} x=_{\e\gamma}O^{\pi'}y\in\B$.  (NExp) of $+_\e$ gives
	$$\{O^\pi x=_{\e\gamma}O^\pi y, O^{\pi'} x=_{\e\gamma}O^{\pi'} y\}\vdash $$
	$$\vdash O^\pi x+_\e O^{\pi'}x=_{\gamma\e}O^\pi y+_\e O^{\pi'}y\in\B.$$
	Hence, $x=_\e y\vdash O^\pi x+_\e O^{\pi'}x=_{\gamma\e}O^\pi y+_\e O^{\pi'}y\in\B$
	\\i.e., $x=_\e y\vdash O^{\e\pi+(1-\e)\pi'} x=_{\gamma\e}O^{\e\pi+(1-\e)\pi'} y\in\B$. 
\end{proof}

Since our working hypothesis is that $\gamma<1$, the previous lemma ensures that
in the fixed-point extension of $\B$, which is $\B^\mu$, we have judgements
involving $\mu x.O^\pi x$.  We use this to show how the Bellman equation can be
solved.

Recall that $[O^\pi]_1^k(s)$ represents the k-th iteration of $O^\pi$ on
$s$. Since $O^\pi$ has only one variable, we drop the lower
index $1$ and write $[O^\pi]^k(s)$ for the k-th iteration on $s$. 

The next theorem is a direct consequence of the Corollary \ref{mu-limit} and Theorem \ref{Bellman-op}.

\begin{theorem}[Bellman equation]\label{Bellman-eq} 
  For any $\pi\in\Pi$ and any $s\in\hat\Sigma^\mu X$, the sequence
  $([O^\pi]^k(s))_{k\geq 1}$ is convergent in $\B^\mu$ and its limit is
  $\mu x.O^\pi x$, i.e., $\forall\e>0 ~\exists n~\forall m$,
  $$\vdash [O^\pi]^{m+n}(s)=_\e\mu x.O^\pi x.$$ Moreover, $\brck{\mu x.O^\pi x}$
  is the unique solution of Bellman equation $$X=T^\pi X.$$
\end{theorem}

Note that the fixed-point Banach theory $\B^\mu$ gives us not only
the solution to Bellman equation, but the apparatus for controlling "the speed" of
convergence of the iteration sequence to the solution of Bellman equation. In
this way, we can build an approximation theory directly inside $\B^\mu$.


\section{Conclusions and related work}
We have developed a quantitative fixed point theory extending the quantitative
equational logic of~\cite{Mardare16} by introducing fixed point operators and
appropriate axioms.  The key ingredients needed were the Banach patterns
that capture the contractiveness of functions in their different arguments.  We
were able to mimic, in this setting, the standard iteration theories as described
in~\cite{Bloom93} and \cite{Simpson00}.  We also developed an extended example
showing that the notion of Bellman equations, which are the centrepiece of
reinforcement learning, can be described in our framework.

A very general and interesting categorical treatment of iteration comes from the
theory of traced monoidal categories~\cite{Joyal96a}.  Recent work by Goncharov
and Schr\"oder~\cite{Goncharov18} develops the notion of guarded traced
categories which, like our Banach patterns, controls when traces can be taken.
The monumental treatise of Bloom and Esik~\cite{Bloom93} also gives a very
general treatment of iteration and mentions fixed points in metric spaces as an
example.  However, these theories are all in the traditional setting of
equational logic and do not have the quantitative notions that we have here with
approximate equality.  Thus, for example, we can discuss the geometric rate of
convergence in value iteration.

A very interesting formulation of the coinduction principle due to
Dexter Kozen \cite{Kozen06,Kozen07} is closely related to our rule for reasoning
about fixed points.  It is equivalent in power to our fixed-point approximation
axiom, as we have argued.  However his rule is very flexible and probably more
convenient to use in various situations.  It would certainly make an interesting
variation to our formulation.  We did consider both alternatives when we were
developing our framework and at the moment we do not see a compelling reason to
choose one over the other.  This is definitely a topic which should be explored
further.

While the fixed-point theory in this paper is infinitary, it would be interesting, as well as potentially useful, to develop a finitary version of it, and in this context, the Kozen principle of coinduction may be more powerful.

We have developed an example showing that some nontrivial situations can be
modelled and reasoned about in our framework.  Of course, whatever we have shown
about Bellman equations has been long known, but it does show the potential power
of the framework.  In recent work Amortila et al.~\cite{Amortila20} have proven
convergence, using coupling techniques, of a variety of more recent
reinforcement learning algorithms.  It would be fascinating to see if the
present framework could help to organize and reason about situations where the
convergence has not yet been established.



%

\bibliography{main}
\end{document}

%% file: macros_prakash.tex
\makeatletter
\newcommand{\singlespacing}{\let\CS=\@currsize\renewcommand{\baselinestretch}{1}\small\CS}
\newcommand{\doublespacing}{\let\CS=\@currsize\renewcommand{\baselinestretch}{1.75}\small\CS}
\newcommand{\normalspacing}{\let\CS=\@currsize\renewcommand{\baselinestretch}{\BLS}\small\CS}
\makeatother

\newtheorem{thm}{Theorem}[section]

\newtheorem{remark}[thm]{Remark}
 
\newtheorem{theorem}[thm]{Theorem}
\newtheorem{corollary}[thm]{Corollary}
\newtheorem{lemma}[thm]{Lemma}
\newtheorem{proposition}[thm]{Proposition}

\newtheorem{df}[thm]{Definition}
\newtheorem{definition}[thm]{Definition}
\newtheorem{example}[thm]{Example}

\newcommand{\reals}{\mathbb{R}}

 \def\pushright#1{{
    \parfillskip=0pt            
    \widowpenalty=10000         
    \displaywidowpenalty=10000  
    \finalhyphendemerits=0      
    \leavevmode                 
    \unskip                     
    \nobreak                    
    \hfil                       
    \penalty50                  
    \hskip.2em                  
    \null                       
    \hfill                      
    {#1}                        
    \par}}                      
 \def\qed{\pushright{\rule{2mm}{3mm}}\penalty-700 \smallskip}

\newenvironment{proof}{\begin{trivlist} \item[{\bf ~Proof}.]}%
 {\qed\end{trivlist}}

\makeatletter

\newdimen\w@dth

\def\setw@dth#1#2{\setbox\z@\hbox{\scriptsize $#1$}\w@dth=\wd\z@
\setbox\@ne\hbox{\scriptsize $#2$}\ifnum\w@dth<\wd\@ne \w@dth=\wd\@ne \fi
\advance\w@dth by 1.2em}

\def\t@^#1_#2{\allowbreak\def\n@one{#1}\def\n@two{#2}\mathrel
{\setw@dth{#1}{#2}
\mathop{\hbox to \w@dth{\rightarrowfill}}\limits
\ifx\n@one\empty\else ^{\box\z@}\fi
\ifx\n@two\empty\else _{\box\@ne}\fi}}
\def\t@@^#1{\@ifnextchar_ {\t@^{#1}}{\t@^{#1}_{}}}

\def\t@left^#1_#2{\def\n@one{#1}\def\n@two{#2}\mathrel{\setw@dth{#1}{#2}
\mathop{\hbox to \w@dth{\leftarrowfill}}\limits
\ifx\n@one\empty\else ^{\box\z@}\fi
\ifx\n@two\empty\else _{\box\@ne}\fi}}
\def\t@@left^#1{\@ifnextchar_ {\t@left^{#1}}{\t@left^{#1}_{}}}

\def\two@^#1_#2{\def\n@one{#1}\def\n@two{#2}\mathrel{\setw@dth{#1}{#2}
\mathop{\vcenter{\hbox to \w@dth{\rightarrowfill}\kern-1.7ex
                 \hbox to \w@dth{\rightarrowfill}}%
       }\limits
\ifx\n@one\empty\else ^{\box\z@}\fi
\ifx\n@two\empty\else _{\box\@ne}\fi}}
\def\tw@@^#1{\@ifnextchar_ {\two@^{#1}}{\two@^{#1}_{}}}

\def\tofr@^#1_#2{\def\n@one{#1}\def\n@two{#2}\mathrel{\setw@dth{#1}{#2}
\mathop{\vcenter{\hbox to \w@dth{\rightarrowfill}\kern-1.7ex
                 \hbox to \w@dth{\leftarrowfill}}%
       }\limits
\ifx\n@one\empty\else ^{\box\z@}\fi
\ifx\n@two\empty\else _{\box\@ne}\fi}}
\def\t@fr@^#1{\@ifnextchar_ {\tofr@^{#1}}{\tofr@^{#1}_{}}}

\newdimen\W@dth
\def\setW@dth#1#2{\setbox\z@\hbox{$#1$}\W@dth=\wd\z@
\setbox\@ne\hbox{$#2$}\ifnum\W@dth<\wd\@ne \W@dth=\wd\@ne \fi
\advance\W@dth by 1.2em}

\def\T@^#1_#2{\allowbreak\def\N@one{#1}\def\N@two{#2}\mathrel
{\setW@dth{#1}{#2}
\mathop{\hbox to \W@dth{\rightarrowfill}}\limits
\ifx\N@one\empty\else ^{\box\z@}\fi
\ifx\N@two\empty\else _{\box\@ne}\fi}}
\def\T@@^#1{\@ifnextchar_ {\T@^{#1}}{\T@^{#1}_{}}}

\def\T@left^#1_#2{\def\N@one{#1}\def\N@two{#2}\mathrel{\setW@dth{#1}{#2}
\mathop{\hbox to \W@dth{\leftarrowfill}}\limits
\ifx\N@one\empty\else ^{\box\z@}\fi
\ifx\N@two\empty\else _{\box\@ne}\fi}}
\def\T@@left^#1{\@ifnextchar_ {\T@left^{#1}}{\T@left^{#1}_{}}}

\def\Tofr@^#1_#2{\def\N@one{#1}\def\N@two{#2}\mathrel{\setW@dth{#1}{#2}
\mathop{\vcenter{\hbox to \W@dth{\rightarrowfill}\kern-1.7ex
                 \hbox to \W@dth{\leftarrowfill}}%
       }\limits
\ifx\N@one\empty\else ^{\box\z@}\fi
\ifx\N@two\empty\else _{\box\@ne}\fi}}
\def\T@fr@^#1{\@ifnextchar_ {\Tofr@^{#1}}{\Tofr@^{#1}_{}}}

\def\Two@^#1_#2{\def\N@one{#1}\def\N@two{#2}\mathrel{\setW@dth{#1}{#2}
\mathop{\vcenter{\hbox to \W@dth{\rightarrowfill}\kern-1.7ex
                 \hbox to \W@dth{\rightarrowfill}}%
       }\limits
\ifx\N@one\empty\else ^{\box\z@}\fi
\ifx\N@two\empty\else _{\box\@ne}\fi}}
\def\Tw@@^#1{\@ifnextchar_ {\Two@^{#1}}{\Two@^{#1}_{}}}

\def\to{\@ifnextchar^ {\t@@}{\t@@^{}}}
\def\from{\@ifnextchar^ {\t@@left}{\t@@left^{}}}
\def\two{\@ifnextchar^ {\tw@@}{\tw@@^{}}}
\def\tofro{\@ifnextchar^ {\t@fr@}{\t@fr@^{}}}
\def\To{\@ifnextchar^ {\T@@}{\T@@^{}}}
\def\From{\@ifnextchar^ {\T@@left}{\T@@left^{}}}
\def\Two{\@ifnextchar^ {\Tw@@}{\Tw@@^{}}}
\def\Tofro{\@ifnextchar^ {\T@fr@}{\T@fr@^{}}}

\makeatother

%% file: IterationTheory.bbl
\begin{thebibliography}{APPB20}

\bibitem[AJV96]{Joyal96a}
Ross~Street Andre~Joyal and Dominic Verity.
\newblock Traced monoidal categories.
\newblock {\em Math. Proc. Camb. Phil. Soc.}, 119:447--468, 1996.

\bibitem[APPB20]{Amortila20}
Philip Amortila, Doina Precup, Prakash Panangaden, and Marc Bellemare.
\newblock A distributional analysis of sampling-based reinforcement learning
  algorithms.
\newblock In {\em The 23rd International Conference on Artificial Intelligence
  and Statistics}, 2020.

\bibitem[Bak71]{DeBakker71}
Jaco W.~De Bakker.
\newblock {\em Recursive procedures}.
\newblock Number~24 in Mathematical Centre Tracts. Mathematisch Centrum,
  Amsterdam, 1971.

\bibitem[Ban22]{Banach22}
Stefan Banach.
\newblock Sur les op\'erations dans les ensembles abstraits et leur application
  aux \'equations int\'egrales.
\newblock {\em Fundamenta Mathematicae. 3: 133–181}, 3:133--181, 1922.

\bibitem[BE93]{Bloom93}
S.~Bloom and Z.~\'Esik.
\newblock Iteration theories.
\newblock {\em EATCS Monographs on Theoretical Computer Science}, 1993.

\bibitem[EB95]{Esik95}
Z.~Esik and L.~Bernatsky.
\newblock {Scott Induction and Equational Proofs}.
\newblock {\em ENTCS}, 1:154--181, 1995.

\bibitem[GS18]{Goncharov18}
Sergey Goncharov and Lutz Schr{\"{o}}der.
\newblock Guarded traced categories.
\newblock In Christel Baier and Ugo~Dal Lago, editors, {\em Foundations of
  Software Science and Computation Structures - 21st International Conference,
  {FOSSACS} 2018, Held as Part of the European Joint Conferences on Theory and
  Practice of Software, {ETAPS} 2018, Thessaloniki, Greece, April 14-20, 2018,
  Proceedings}, volume 10803 of {\em Lecture Notes in Computer Science}, pages
  313--330. Springer, 2018.

\bibitem[Has99]{Hasegawa99}
M.~Hasegawa.
\newblock {\em Models of Sharing Graphs: A Categorical Semantics of let and
  letrec.}
\newblock Distinguished Dissertation Series, Springer-Verlag, 1999.

\bibitem[Kle52]{Kleene52}
S.~C. Kleene.
\newblock {\em Introduction to metamathematics}.
\newblock North-Holland, Amsterdam, 1952.

\bibitem[Koz06]{Kozen06}
Dexter Kozen.
\newblock Coinductive proof principles for stochastic processes.
\newblock In Rajeev Alur, editor, {\em Proceedings of the 21st Annual IEEE
  Symposium On Logic In Computer Science LICS'06}, pages 359--366, August 2006.

\bibitem[Koz07]{Kozen07}
Dexter Kozen.
\newblock Coinductive proof principles for stochastic processes.
\newblock {\em Logical Methods In Computer Science}, 3(4:8):1--14, 2007.

\bibitem[MPP16]{Mardare16}
Radu Mardare, Prakash Panangaden, and Gordon Plotkin.
\newblock Quantitative algebraic reasoning.
\newblock In {\em Proceedings of the 31st Annual ACM-IEEE Symposium on Logic in
  Computer Science}, pages 700--709, 2016.

\bibitem[MPP17]{Mardare17}
Radu Mardare, Prakash Panangaden, and Gordon Plotkin.
\newblock On the axiomatizability of quantitative algebras.
\newblock In {\em Proceedings of the 32nd Annual ACM-IEEE Symposium on Logic in
  Computer Science}, 2017.

\bibitem[Put94]{Puterman94}
Martin~L. Puterman.
\newblock {\em Markov {D}ecision {P}rocesses: {D}iscrete Stochastic Dynamic
  Programming}.
\newblock Wiley, 1994.

\bibitem[SB]{Scott69}
Dana Scott and Jaco W.~De Bakker.
\newblock A theory of programs.
\newblock Unpublished notes, IBM Seminar, Vienna.

\bibitem[SB98]{Sutton98}
Richard~S. Sutton and Andrew~G. Barto.
\newblock {\em Reinforcement Learning: {A}n Introduction}.
\newblock MIT Press, 1998.

\bibitem[SP00]{Simpson00}
Alex Simpson and Gordon Plotkin.
\newblock Complete axioms for categorical fixed-point operators.
\newblock In {\em Proceedings of the 15th Annual IEEE Symposium on Logic in
  Computer Science (LICS 2000)}, pages 30--41. IEEE, June 2000.

\end{thebibliography}
